\documentclass{book}

\usepackage{makeidx}
\usepackage{amsthm,amsmath,amssymb,amsfonts}
\usepackage{setspace,graphicx}
\usepackage{Generic}  % keep the Generic.sty file in the same folder
\usepackage{url} 
\usepackage{mdframed}
\newtheoremstyle{break}%
{}{}%
{\normalfont}{}%
{\bfseries}{}% % Note that final punctuation is omitted.
{\newline}{}

\newtheorem{lemma}{Lemma}
\theoremstyle{plain}
\newtheorem{theorem}{Theorem}[section]
\theoremstyle{break}
\newtheorem{example}{Example}[section]

\surroundwithmdframed[
  linewidth=2pt,
  linecolor=black,
  topline=false,
  bottomline=false,
  rightline=false,
  leftline=true,
  innerleftmargin=5pt,
  innerrightmargin=5pt,
  skipabove=10pt,
  skipbelow=10pt
]{example}

\surroundwithmdframed[
  linewidth=2pt,
  linecolor=black,
  topline=false,
  bottomline=false,
  rightline=false,
  leftline=true,
  innerleftmargin=5pt,
  innerrightmargin=5pt,
  skipabove=10pt,
  skipbelow=10pt
]{theorem}

\begin{document}

%\doublespacing
%  The chapter command will create the title heading for your chapter
%  
%  Put the title in the braces { } and put a short ``running head'' 
%  version of the title (an abbreviated version of the title) to 
%  appear at the top of the left hand pages in the brackets

%  NOTE:  only the first word in the title starts with a capital letter!
%  Same for the short title

\setcounter{chapter}{18} % Set chapter number to 19

\chapter[Perturbations of Markov Chains]{Perturbations of Markov Chains}
\vspace{-12mm}
\begin{Large}
{\em  Daniel Rudolf, Aaron Smith, and Matias Quiroz\footnote{To appear as Chapter 19 in the second edition of the Handbook of MCMC.}}
\end{Large}

\bigskip
\bigskip
\bigskip
\bigskip
\bigskip

\bigskip

Running MCMC involves compromise. At a minimum, we must almost always replace the ideal continuous state space of our statistical models with the discrete state spaces that can fit in compute memory. More sophisticated algorithms, such as those based on Langevin diffusion or Hamiltonian Monte Carlo, require further approximations. For more complex models and larger datasets, even evaluating the likelihood at a single point may be intractable, and so exact likelihood evaluations may be replaced by surrogates and approximations. 

These compromises introduce bias to the results of the MCMC computations. While some of these biases can be corrected easily (e.g. by adding a Metropolis-Hastings adjustment to the Langevin algorithm as in \cite{Besag94}), others cannot. In these situations, a careful analysis of the algorithm requires analysis of this non-vanishing bias. Perturbation theory refers to the simplest and most widespread approach to bounding this bias. The basic idea is to compare the perturbed MCMC algorithm that is actually being run to the ideal MCMC algorithm that you would like to run. As long as the chains are sufficiently close (and other technical conditions are met), you can guarantee that the bias is small. Indeed, exactly this approach has been used to show that the fine discretizations of continuous state-space models used on computers introduce a negligible amount of bias \cite{BREYER2001123}.

This chapter surveys progress on three related topics: the motivating question of when and how ``perturbed" MCMC chains are developed, the theoretical problem of how perturbation theory can be used to analyze such chains, and finally the question of how the theoretical analyses can lead to practical advice.

\section{Introduction to Perturbed MCMC: Using What You Know}\label{intro}%

As we've seen throughout this book, the choice of Monte Carlo algorithm is often driven by how much the users know about the target distribution. If they know a great deal about the structure of the target, they might use a specialized algorithm such as the Polya-Gamma sampler of \cite{PolsonPolyaGamma12}. In the prototypical case that they can evaluate the target density up to an unknown multiplicative constant, they might use the Metropolis-Hastings algorithm. If they can't even do that, but know how to compute an unbiased estimator of the target density, they might use pseudomarginal MCMC (see Chapter 16). 

While these algorithms have wildly different behaviours, they are all ``asymptotically exact" algorithms - that is, basic MCMC theory guarantees that the associated Markov chain will eventually converge to the target distribution under modest conditions \cite{meyn_tweedie_glynn_2009}. A typical analysis of the efficiency of an asymptotically exact algorithm will focus on its asymptotic variance, as in \cite{MCCLT04}.

This chapter is about the analysis of algorithms that aren't even asymptotically exact. As a typical example, the approximate Bayesian computation algorithm is a popular choice when the target distribution is not even known up to a normalizing constant, but one can produce exact samples from the data generating process of the associated statistical models \cite{Tav97}. These algorithms will typically have \textit{some} stationary distribution, but it will not be the ostensible target. As a result, a theoretical analysis requires bounds both on the asymptotic variance \textit{and} on the bias. The main idea behind perturbation theory in the context of MCMC is to analyze the bias of a non-exact algorithm by comparing it to a ``similar" exact algorithm. Just as the choice of exact MCMC algorithm is heavily guided by what is easy to compute, so is the choice of perturbed MCMC algorithm. 

\section{Basic Principles}

We describe the basic principles of perturbed Markov chains, focusing on toy examples.

\subsection{Basic Principles: A Taxonomy of Perturbed Markov Chains} \label{SecTaxonomy}

We give a brief and highly subjective taxonomy of situations in which perturbed Markov chain methods have been useful. First, we divide methods by the type of problem that they are intended to solve:

\begin{enumerate}
    \item \textbf{The target distribution is computationally intractable.} For example, the target distribution might be defined in terms of the solutions of expensive-to-solve PDE models. Papers such as \cite{ApproxBayes10} and \cite{Conrad16} deal with this by reducing the number of times a PDE is solved or the fidelity of the solutions. More simply, evaluating the target distribution may involve reading a large dataset.
    
    \item \textbf{The ideal algorithm is computationally intractable.} For example, Hamiltonian Monte Carlo \cite{DUANE1987216,Neal95} involves solving a system of differential equations. In practice, these differential equations are solved by numerical integrators such as the Verlet integrator \cite{Neal95}. Typical analyses, such as \cite{pmlr-v89-mangoubi19a,Chen20}, depend heavily on careful comparisons between ``ideal" and ``perturbed" HMC algorithms. 
    \item \textbf{The nominal target distribution is not the real target distribution.} For example, subsampling and other approximations can be used to obtain ``implicit" regularization \cite{NEURIPS2021_2b6bb535, smith2021regularisationSGD}, tempering \cite{Meent14}, or other desirable adjustments to the target. Although the actual algorithms in this category may be exactly the same as those in the first category, the viewpoint is very different.
\end{enumerate}

The focus of this chapter is on the first problem, though perturbation techniques are also used for the other two. We give an expanded list of some common reasons that one might use a perturbed MCMC algorithm for the first problem. These are ordered very roughly by the amount of information required to use the method: proxies near the start make strong assumptions and tend to be easier to use efficiently, while proxies near the end assume much less and tend to have more severe problems.

\begin{enumerate}
\item \textbf{Subsampling and Approximations for Simple Targets:} The posterior of simple models such as generalized linear models and other ``product-form" likelihoods is
\begin{equation} \label{EqProdPost}
\pi(\theta) =\frac{f(\theta | X)}{\int f(\theta | X)d\theta},\quad f(\theta | X) = f(\theta) \prod_{i=1}^{n} f(X_{i}|\theta),
\end{equation}
where $f(\theta)$ is the prior. We refer to  $f(\theta|X)$ in \eqref{EqProdPost} as the posterior (ignoring the normalizing constant).  For very large datasets, the posterior $f(\theta|X)$ is difficult to evaluate. In this context, it is natural to use only a sample of the data when computing posterior distributions. At its simplest, this would just mean choosing a random sample $S \subseteq \{1,2,\ldots,n\}$ and using a ``plug-in" estimator of the form
\begin{equation} \label{EqNoProxy}
\widehat{f}(\theta | X) = f(\theta) \prod_{i \in S} f(X_{i}|\theta),
\end{equation}
whenever one would usually use the full posterior in \eqref{EqProdPost}. In practice, this simple posterior is often too crude to yield a ``small" perturbation of the original algorithm \cite{bardTall17}. Intuitively, this random posterior estimate is highly variable, as it is based on a small number of observations. This may be addressed by pre-computing a global surrogate function $G$ (that accounts for the observations that were not sampled) and using surrogates of the form
\begin{equation} \label{EqSimpleProxy}
\widehat{f}(\theta | X) = f(\theta) \prod_{i \in S} f(X_{i}|\theta)^{\alpha} G(\theta)^{\beta},
\end{equation}
with suitably chosen $\alpha,\beta\in \mathbb{R}$.

There is a very large literature dedicated to improving this simple idea. See e.g. \cite{wellingSGLD11} for a very influential early paper, \cite{Korattikara2013AusterityIM, bardenet2014towards} for early applications of subsampling in the Metropolis-Hastings algorithm, \cite{Quiroz19,Quiroz18} for practical advice on pre-computation in subsampling methods, \cite{Montanari14} for an illustration of how natural pre-computations can fail completely, and  \cite{Core16, Chen22} for a distinctive approach based on reducing a large dataset to a carefully chosen and logarithmically-small ``core" set. 

\item \textbf{Parallelization and Divide-And-Conquer:} Sometimes, due to either privacy concerns or limitations on machine memory, it is not practical to hold an entire dataset in one computer. In this situation, it is not possible to run standard subsampling algorithms. This leads to ``divide-and-conquer" or ``parallel" algorithms. Conceptually, these algorithms have two parts: an MCMC algorithm that is run on each individual computer, and an overall algorithm that merges the MCMC outputs. The simplest of these algorithms runs independent MCMC algorithms on each computer and then merge the full outputs in a postprocessing step, as in \cite{scott2022bayes} and \cite{Wang2013ParallelizingMV}. More complex algorithms allow a small amount of communication between computers while the MCMC chains are running, as in \cite{neiswanger2013asymptotically,EPWay2020,RendelGCMC21}.

\item \textbf{Subsampling and Approximations for Complex Targets} 

The product-form distributions in equation \eqref{EqProdPost} differ from more complex target distributions in two important ways. The first is that it is very easy to write down a sensible plug-in estimator (as in equations \eqref{EqNoProxy} and \eqref{EqSimpleProxy}) for product-form targets, while there may not be any straightforward plug-in estimator for more complex targets. The second is that the computational complexity of evaluating product-form targets scales roughly linearly in the size of the dataset\footnote{Ignoring discontinuities that come from e.g. the sizes of different types of machine memory.} $n$, while more complex targets may have complexity that grows much more quickly (e.g. if matrix operations are required as in \cite{Horseshoe2020,rastelli2023computationally}). As a consequence, subsampling is both more valuable and more difficult in this context.

As a concrete example of this phenomenon, evaluating the posterior distribution for latent position models of random graphs has a cost that grows like the square of the size of the dataset, but the posterior has a product structure that can easily be approximated \cite{rastelli2023computationally}.

\item \textbf{Inverse Problems from PDEs and Other Expensive Models:} Models such as \eqref{EqProdPost} are amenable to subsampling or divide-and-conquer because of two important properties: (i) it is easy to relate the posterior distribution of a \textit{full} sample and a \textit{partial} sample, and (ii) the partial-sample posteriors are much easier to compute. In the previous point, we saw that it was sometimes possible to extend \textit{similar} ideas to more complex models such as \cite{Horseshoe2020}. However, at some point, this approach seems to break down completely. 

As a prototypical example, consider the problem of making inferences about subsurface hydrology based on surface measurements such as electromagnetic surveys \cite{Subwater10}. The parameters of interest are related to the observations by the solution to a differential equation, and there is no obvious way to substantially speed up this computation by looking at only subsamples of the data.\footnote{Subsampling may substantially speed up computation for other reasons - for example, the associated posterior distribution may be simpler and thus easier to sample from.} In this context, one typically treats the likelihood as a black box and focuses on reducing the number of evaluations \cite{Conrad16} or their fidelity \cite{DodSub15}. Note that this is closely related to delayed-acceptance MCMC, which is \textit{unperturbed} but typically gives smaller per-step savings (see e.g. \cite{MLMDA23} for a sophisticated version of delayed-acceptance MCMC in the present context). 
\item \textbf{Simulation-Based Methods:} In some models, there is no tractable way to evaluate the likelihood even once - but it is possible to \textit{simulate} from the model for fixed parameter values. As a prototypical toyexample, consider a simple model of the spread and evolution of a disease in a large population, but with only one datapoint collected: the haplotype of the disease of the last infected person. While it is straightforward to sample from many epidemic models, and it is straightforward to calculate the likelihood of any \textit{full} path of an epidemic, calculating the likelihood of only \textit{one person} would involve a completely intractable summation over the full unobserved infection path. A common approach in this setting is approximate Bayesian computation (see the main example in \cite{Jarno16})\footnote{Exact approaches are also common - the simplest is to augment the state space with the missing data - but we don't discuss them in this chapter.}. See Chapter 16 of \cite{brooks2011handbook}, Chapter 11 of this book, and \cite{Dashti2017} for more details.
\item \textbf{Doubly-Intractable Distributions:} Finally, there are the worst problems: those for which you can neither evaluate the likelihood even once, nor easily simulate from the model for fixed parameter values. A standard simple example is the \textit{exponential random graph model} (ERGM), an early version of which was proposed in \cite{ERGM96}. These models have densities of the form
\begin{equation} 
f(G|\theta) = Z_{\theta}^{-1} e^{\sum_{j=1}^{k} \theta_{j} S_{H_{j}}(G)},
\end{equation} 
where $G$ is a graph on $n$ vertices, $\theta =(\theta_{1},\ldots,\theta_{k})$ are parameters, $H_{1},\ldots,H_{k}$ are graphs with fewer than $n$ vertices, and $S_{H}(G)$ is the number of times that $H$ appears as a subgraph in $G$. Note that the normalizing constant $Z_{\theta}$ depends on the parameter $\theta$ and cannot generally be computed, so standard MCMC algorithms cannot be applied to estimate $\theta$. Even worse, sampling from the model for fixed $\theta$ is usually difficult and sometimes intractable \cite{Bhamidi2008MixingTO}.

The survey \cite{IntractableSurvey18} covers a wide variety of approaches for these very difficult problems. Some use only a very limited number of simulations (e.g. the exchange algorithm \cite{MurraryExchange06}, the auxiliary-variable methods \cite{MollerIntr06}, and their approximations \cite{kang2023measuring}), a moderate number of simulations (e.g. \cite{ERGM23}), or none at all (e.g. \cite{Roulette15}). 

On seeing these algorithms for the first time, it is natural to ask: since it is easy to obtain an unbiased estimate of the normalizing constant $Z_{\theta}$, could one not use a trick similar to the pseudomarginal algorithm to obtain exact MCMC? The answer is unfortunately a resounding \textit{no}: under very general conditions, it is impossible to turn an unbiased estimate of $Z_{\theta}$ into a nonnegative unbiased estimator of $Z_{\theta}^{-1}$. For details we refer to \cite{jacob2015nonnegative}.

\end{enumerate}

Note that we don't discuss theoretical approaches, such as the ``scaling limit" approach introduced in \cite{gelman97}, that don't directly control bias. We also don't discuss algorithms, such as \cite{Besag94, Crist05, quiroz18delayed}, that can be viewed as approximations of ``ideal" algorithms but still maintain the exact target distribution.

\subsection{Basic Principles: Mathematical Theory}
\label{subsec: basic_math_theory}

We provide a gentle introduction into the mathematical theory, including a simple theorem and two applications. These results are all pedagogical - see Section~\ref{SecMathPert} for a  more powerful theorem and Section~\ref{SecExampleAppsGeneric} for more realistic applications.

We begin by setting the notation. Let $\Omega$ be a Polish space and let $\mathcal{F}=\mathcal{B}(\Omega)$ be the associated Borel $\sigma$-algebra. Define $\Pi(\mu,\nu)$ to be the collection of all \textit{couplings} of distributions $\mu, \nu$ on $\Omega$, that is, the set of probability measures $\xi$ on $(\Omega\times\Omega,\mathcal{F}\otimes \mathcal{F})$
satisfying
\begin{equation*}\xi(A,\Omega) = \mu(A), \qquad \xi(\Omega,A) = \nu(A)
\end{equation*}
for all $A \in \mathcal{F}$. 
Let $d$ be a metric\footnote{It is implicitly assumed that $d\colon \Omega\times \Omega \to [0,\infty)$ is lower semi-continuous w.r.t. the product topology of $\Omega$.} on $\Omega$.
A common tool to measure the difference of distributions in the setting of perturbation theory of Markov chains is the \textit{Wasserstein distance}, defined as 
\begin{equation*}
W_{d}(\mu,\nu) = \inf_{\xi \in \Pi(\mu,\nu)} \mathbb{E}_{(X,Y) \sim \xi}[d(X,Y)].
\end{equation*}
The popular total variation distance is given by the Wasserstein distance associated with the trivial metric $d(x,y) = \textbf{1}_{x \neq y}$, cf. \cite[Theorem~19.1.6]{douc2018markov}.  We use $d_{TV}$ as shorthand for this metric on probability spaces and mention already that its deficiency is that it is not able to take a possible ``metric fine structure" of the space into account.

We give our first perturbation result. Let $(X_n)_{n\in\mathbb{N}_0}$ be an unperturbed, exact Markov chain with transition kernel $Q$ on $\Omega$ and $(\widetilde{X}_n)_{n\in\mathbb{N}_0}$ be a perturbed one with transition kernel $K$ on $\Omega$. By $p_n$ and $\widetilde{p}_n$ we denote the distribution of $X_n$ and $\widetilde{X}_n$, respectively. Moreover, remind yourself that $p_n=p_0Q^n$ and $\widetilde{p}_n= \widetilde{p}_0 K^n$.

We make two main assumptions. First, the transition kernel of the unperturbed chain satisfies a contraction property that yields exponential convergence to the corresponding stationary distribution. Second, we impose a uniform approximation property, that gives that the perturbed and the unperturbed kernel are close to each other. More precisely:
\begin{enumerate}
\item \textbf{Contraction Property:} There exists $\alpha > 0$ so that 
\begin{equation}\label{IneqSimpleCont}
W_{d}(Q(x,\cdot), Q(y,\cdot))\leq (1-\alpha)d(x,y), \qquad \forall x,y \in \Omega.
\end{equation}
%\item \textbf{Approximation Property:} %\textcolor{red}{$\beta$ should be $\delta$, at least in the proof below} 
%There exists $0 \leq \beta < \infty$ so that for any $x,y \in \Omega$,
%\begin{equation}\label{IneqSimpleApprox}
%W_{d}(Q(x,\cdot), K(y,\cdot))\leq  d(x,y) + \beta.
%\end{equation}
\item 
\textbf{Approximation Property:} There exists $0\leq \varepsilon <\infty$ so that 
\begin{equation} \label{IneqSimpleApprox}
    W_d(Q(x,\cdot),K(x,\cdot)) \leq \varepsilon, \qquad \forall x\in\Omega.
\end{equation}

\end{enumerate}
We formulate a first simple statement about the difference of the perturbed and unperturbed marginal $n$-step distributions.
\begin{theorem}\label{1th:Z_m}
Under Assumptions \eqref{IneqSimpleCont} and \eqref{IneqSimpleApprox} we have
% \begin{equation} \label{eq: 1st_pert_est}
% W_{d}(K^{n}(x,\cdot), Q^{n}(y,\cdot)) \leq (1-\alpha)^{n} d(x,y) + n \varepsilon,
% \end{equation}
\begin{equation} \label{eq: 1st_pert_est}
W_{d}(p_n, \widetilde{p}_n) \leq (1-\alpha)^{n} W_d(p_0,\widetilde{p}_0) + \varepsilon\, \min\{n,\alpha^{-1}\},
\end{equation}
for any $n\in\mathbb{N}$.
% , $x,y\in \Omega$.
In particular, for stationary distributions $\pi_Q,\pi_K$ of $Q,K$, with $W_d(\pi_Q,\pi_K)<\infty$, this implies 
\begin{equation} \label{eq: 1st_pert_stat}
W_{d}(\pi_Q,\pi_K) \leq  \frac{ \varepsilon}{\alpha}.
\end{equation}
\end{theorem}

The proof is based on the simplest possible perturbation argument: coupling entire \textit{paths} of two Markov chains so that they remain close for long periods.

\begin{proof}
It is well known, see for example \cite[Theorem~20.1.1.]{douc2018markov}, that there is
an optimal coupling $\eta\in \Pi(p_0,\widetilde{p}_0)$, such that
\[
 W_d(p_0,\widetilde{p}_0) = \int_{\Omega\times\Omega} d(x,y) \eta({\rm d}x,{\rm d}y).
\]
Moreover, for any $x,y\in \Omega$ there is a transition kernel $\xi\colon \Omega^2 \times \mathcal{F}\otimes \mathcal{F} \to [0,1]$, such that $\xi(x,y,\cdot) \in \Pi(K(x,\cdot),Q(y,\cdot))$ is also an optimal coupling, i.e., 
\[
W_d(K(x,\cdot),Q(y,\cdot)) = \int_{\Omega\times \Omega} d(z_1,z_2)\,\xi(x,y,{\rm d}z_1,{\rm d}z_2),
\]
see for example \cite[Theorem~1.1]{Zh00}.  Consider a Markov chain $(X_n,\widetilde{X}_n)_{n\in\mathbb{N}_0}$ on $\Omega\times \Omega$ with transition kernel $\xi$ and initial distribution $\eta$. Observe that $(X_n)_{n\in\mathbb{N}_0}$, $(\widetilde{X}_n)_{n\in\mathbb{N}_0}$ are Markov chains with corresponding kernels $K,Q$ and initial distributions $p_0$, $\widetilde{p}_0$, respectively. Then, for $n\in\mathbb{N}$, we have
$
  W_d(p_n,\widetilde{p}_n) =  W_d(p_0K^n,\widetilde{p}_0Q^n)
   \leq \mathbb{E}(d(X_n,\widetilde{X}_n))
$
and \begin{align*}
   \mathbb{E}(d(X_n,\widetilde{X}_n))
   = &  \mathbb{E} \left[ \mathbb{E}(d(X_n,\widetilde{X}_n) \mid X_{n-1},\widetilde{X}_{n-1}) \right]
  =  \mathbb{E} \left[ W_d(K(X_{n-1},\cdot),Q(\widetilde{X}_{n-1},\cdot))\right]\\ 
  \leq & \mathbb{E} \left[ W_d(K(X_{n-1},\cdot),K(\widetilde{X}_{n-1},\cdot))\right]+\mathbb{E}\left[W_d(K(\widetilde{X}_{n-1},\cdot),Q(\widetilde{X}_{n-1},\cdot))\right]\\
  \leq & (1-\alpha) \mathbb{E}(d(X_{n-1},\widetilde{X}_{n-1})) + \varepsilon,
\end{align*}
where we used \eqref{IneqSimpleCont} and \eqref{IneqSimpleApprox} in the last inequality. Applying the formerly derived estimation iteratively yields
\[
    \mathbb{E}(d(X_n,\widetilde{X}_n)) \leq (1-\alpha)^n  \mathbb{E}(d(X_{0},\widetilde{X}_{0})) + \varepsilon \sum_{i=0}^{n-1} (1-\alpha)^i.
\]
Finally, by $\sum_{i=0}^{n-1} (1-\alpha)^i \leq \min\{n,\alpha^{-1}\}$ and $W_d(p_0,\widetilde{p}_0)=\mathbb{E}(d(X_0,\widetilde{X}_0))$ we obtain \eqref{eq: 1st_pert_est}.

To obtain \eqref{eq: 1st_pert_stat}, we apply \eqref{eq: 1st_pert_est} with $n=1$ taking starting distributions $\pi_Q,\pi_K$ into account and get
\begin{equation*}
W_{d}(\pi_Q,\pi_K) = W_{d}(\pi_Q Q, \pi_K K) \leq (1-\alpha) W_{d}(\pi_Q,\pi_K) + \varepsilon,
\end{equation*}
where rearranging gives \eqref{eq: 1st_pert_stat}.
\end{proof}

Let us note that, in the assumptions of Theorem \ref{1th:Z_m}, the roles of $Q$ and $K$ are interchangeable - the result is also true when $Q$ is the perturbed and $K$ the unperturbed transition kernel.
Qualitatively, the result regarding the difference w.r.t. the stationary distributions is typical in perturbation theory of Markov chains. Similar results can already be found in early papers in discrete state space settings \cite{kart86,Seneta91,mitrophanov2005sensitivity,PertStat13,Kartashov2013MaximalCP,CASWELL20131727} and many later papers on MCMC \cite{johndrow2017error,SchweizerPerturb18,negrea_rosenthal_2021}. We can say perturbation bounds scale \textit{at worst} like the product of the per-step perturbation error (in this case, $\varepsilon$) and the time to get close to stationarity (in this case, $\alpha^{-1}$). 
% , under the condition that this product is sufficiently small.

We now give a prototypical example showing how these bounds can be used and that they are sharp. In particular, we emphasize that the use of the Wasserstein distance instead of the total variation distance is beneficial.

\begin{example}[Random walk taking Euclidean distances into account] \label{ExEucDist}
Assume that there are three locations $x,\widetilde{x},y \in \mathbb{R}$ satisfying $\widetilde{x}\leq x < y$ and a probability parameter $q\in (0,1)$. We consider $\mathbb{R}$ equipped with the Euclidean metric $d$ induced by the absolute value $\vert\cdot\vert$ and set $\Omega := \{x,\widetilde{x},y\}$. The unperturbed transition kernel is given by
\[
Q(z,A) := ((1-q)\, \delta_x(A) + q\, \delta_y(A))\mathbf{1}_{z\in\{x,\widetilde{x}\}} + (q\, \delta_x(A) + (1-q)\, \delta_y(A))\mathbf{1}_{z=y},
\]
where $\delta_{\cdot}$ denotes the Dirac measure, $z\in \Omega$ and $A\subseteq\Omega$. 
Whenever $z\in\{x,\widetilde{x}\}$, then with probability $q$ you stay at or jump to $x$ and with probability $1-q$ you jump to $y$; whenever $z=y$, then with probability $1-q$ you jump to $x$ and with probability $q$ you stay at $y$. Essentially, this is a random walk between $x$ and $y$, with transient state $\widetilde{x}$. In the perturbed kernel we change the roles of $x$ and $\widetilde{x}$ in the Dirac measures, so that
\[
K(z,A) := ((1-q)\, \delta_{\widetilde{x}}(A) + q\, \delta_y(A))\mathbf{1}_{z\in\{x,\widetilde{x}\}} + (q\, \delta_{\widetilde{x}}(A) + (1-q)\, \delta_y(A))\mathbf{1}_{z=y}.
\]
The intuition is as follows: One wants to perform a random walk between $x,y$, but only a proxy from below of $x$, namely $\widetilde{x}$, is available.
Assuming without loss of generality that $\vert x-y \vert \leq \vert \widetilde{x}-y \vert$ and exploiting a well known formula\footnotemark \,  
for the Wasserstein distance on $\mathbb{R}$ one can show
\begin{align*}
    W_{d}(Q(z_1,\cdot),Q(z_2,\cdot)) 
 \leq  \vert 1-2q \vert \vert z_1 - z_2\vert, \quad z_1,z_2\in \Omega.
\end{align*}
Moreover, by similar arguments
\begin{align*}
    W_d(Q(z,\cdot),K(z,\cdot)) \leq \max\{q,1-q\}\vert x-\widetilde{x} \vert, \quad z\in \Omega.
\end{align*}
That is, the ``closer" $x$ and $\widetilde{x}$ are to each other, the ``closer" the kernels $K$ and $Q$ are to each other.
Consequently, by Theorem~\ref{1th:Z_m} we have
\[
W_d(K^n(z_1,\cdot),Q^n(z_2,\cdot)) \leq \vert1-2q\vert^n \vert z_1-z_2 \vert %+ n \max\{q,1-q\}\vert x-\widetilde{x} \vert.
+ \frac{\max\{q,1-q\}}{\vert 1-2q\vert}\vert x-\widetilde{x} \vert.
\]
This readily implies that if the unperturbed and perturbed Markov chain start at the same point, then the difference of the distributions after $n$ steps is essentially bounded by a constant times $\vert x- \widetilde{x}\vert$. 
 Note that $\pi_Q:=1/2 \delta_x + 1/2 \delta_y$, $\pi_K := 1/2 \delta_{\widetilde{x}} + 1/2 \delta_y$ are the corresponding stationary distributions of $Q,K$. For those Theorem~\ref{1th:Z_m} yields
\begin{equation}
\label{eq: 1st_ex_stat_diff}
W_d(\pi_K,\pi_Q) \leq \frac{\max\{q,1-q\}}{1-\vert 1-2q \vert} \vert x-\widetilde{x}\vert.
\end{equation}
Again, if $x$ and $\widetilde{x}$ are ``close" to each other w.r.t. the Euclidean norm, then also the stationary distributions are ``close" to each other.
This is an advantage of the Wasserstein distance (based on $\vert \cdot \vert$) compared to $d_{TV}$. Moreover, by a straightforward direct calculation we have
 $W_d(\pi_K,\pi_Q)= \frac{1}{2} \vert x-\widetilde{x} \vert$, such that for $q=1/2$ the upper bound of \eqref{eq: 1st_ex_stat_diff} cannot be improved.
\end{example}

\footnotetext{For distributions $\mu$, $\nu$ on $\mathbb{R}$ we have $W_d(\mu,\nu)=\int_0^1 \vert F^{-1}_\mu(s)- F^{-1}_\nu(s) \vert {\rm d}s$, where $F^{-1}_\mu,F^{-1}_\nu$ denote the generalized inverse of the CDFs of $\mu,\nu$.}

Now we show that even in easy settings the uniform approximation property \eqref{IneqSimpleApprox} of Theorem~\ref{1th:Z_m} may be too restrictive.

\begin{example}[Restriction of  Theorem~\ref{1th:Z_m} (to be continued)] \label{ExLyap}
Let $\Omega=\mathbb{R}$, $d(x,y)=\vert x-y \vert$ and $\mu$ be a distribution on $\Omega$. For given real-valued random variables $X_0,\widetilde{X}_0$, 
consider the following autoregressive processes 
\begin{align*}
X_{n+1} & =(1-\alpha) X_{n} + Z_n, \qquad
\widetilde{X}_{n+1}  = (1-\widetilde{\alpha}) \widetilde{X}_n + Z_n,
\end{align*}
where $\alpha,\widetilde{\alpha}\in (0,1)$ and $(Z_n)_{n\in\mathbb{N}}$ is an iid sequence with $Z_1\sim \mu$. For the same initial state, the intuition is that as long as $\alpha$ and $\widetilde{\alpha}$ are ``close" to each other,  the Markov chains should also be ``close".
The corresponding transition kernels of the perturbed and unperturbed Markov chain are given as
\[
Q(x,A) = \int_{\mathbb{R}} \mathbf{1}_A(\alpha x+z) \mu({\rm  d}z),
\quad
K(x,A) =  \int_{\mathbb{R}} \mathbf{1}_A(\widetilde{\alpha} x+z) \mu({\rm  d}z),
\]
where $x\in\Omega$ and $A\in\mathcal{B}(\Omega)$. 
For $x,y\in \Omega$, the coupling $$\xi(x,y,{\rm d} \bar{x},{\rm d}\bar{y}) = \int_{\mathbb{R}
} \delta_{(1-\alpha) x + z}({\rm d}\bar{x}) \,\delta_{(1-\alpha) y + z}({\rm d}\bar{y}) \mu({\rm d}z)$$ is in $\Pi(Q(x,\cdot),Q(y,\cdot))$, so
%By a coupling argument w.r.t. $\mu$ one easily can see that
\begin{equation}
\label{contra_restr_ex}
W_d(Q(x,\cdot),Q(y,\cdot)) 
\leq \int_{\mathbb{R}} \vert (1-\alpha) x+z -(1-\alpha) y - z \vert \mu({\rm d}z) = (1-\alpha) 
\vert x- y \vert.
\end{equation}
%with $x,y\in \Omega$.
Therefore \eqref{IneqSimpleCont} is satisfied. Similarly, we obtain
\begin{equation}
\label{upper_err_bnd_restr_ex}
W_d(Q(x,\cdot),K(x,\cdot)) \leq \int_{\mathbb{R}} \vert (1-\alpha) x +z - (1-\widetilde{\alpha}) x - z\vert \mu({\rm d}z) = \vert \alpha - \widetilde{\alpha} \vert \vert x \vert.
\end{equation}
Moreover, by using the Kantorovich-Rubinstein duality we have
\[
W_d(Q(x,\cdot),K(x,\cdot)) = \sup_{\Vert f \Vert_{\rm{Lip}}\leq 1} \vert Qf(x)-Kf(x) \vert, 
\]
with $\Vert f \Vert_{\rm{Lip}}=\sup_{y,z\in \Omega, y\neq z} \frac{\vert f(y)-f(z) \vert}{\vert y-z \vert}$ for $f\colon \Omega \to \mathbb{R}$. Therefore considering $f(x)=x$ yields
\begin{equation}
\label{eq:lower_err_bnd_restr_ex}
W_d(Q(x,\cdot),K(x,\cdot)) \geq \left \vert \int_{\mathbb{R}} (1-\alpha)x+z - (1-\widetilde{\alpha})x -z  \mu({\rm d}z) \right \vert = \vert \alpha-\widetilde{\alpha} \vert \vert x\vert,
\end{equation}
such that \eqref{IneqSimpleApprox} cannot be uniformly in $x$ satisfied and Theorem~\ref{1th:Z_m} is not applicable.
\end{example}

\subsection{Basic Principles: Analyzing a Different Perturbation Assumption} \label{SecDiffPert}

In this chapter, we focus on algorithms that are best analyzed with an approach similar to Theorem \ref{1th:Z_m}. That is, these are algorithms where the best known approach to bounding asymptotic bias is to control the single-step error (as in Assumption \eqref{IneqSimpleApprox}) and add up these errors in some way that is controlled by the mixing behaviour of the chain (as in Assumption  \eqref{IneqSimpleCont}).  The basic problem with this error-propagation approach is that it is difficult to incorporate ``cancellation" of the single-step errors in an analysis. This difficulty can lead to poor estimates, especially when applied naively to multimodal targets with poor mixing properties.

% DR: Otherwise "situation" repeated in the following sentences.
% In some situations, 
Sometimes,
it is possible to analyze error bounds \textit{other than} the single-step error. One desirable approach would be to directly analyze the stationary distribution, and this is easiest when a simple formula for the stationary distribution is available. In some situations, such as approximate Bayesian computation, such a formula will always be available (see \cite{Sisson2018handbook} for details and practical advice). In other situations, it is possible to tweak an existing perturbed algorithm into a very similar one for which an exact formula for the stationary distribution is available. The most prominent such tweak is based on the pseudomarginal algorithm of \cite{pseudomarginal}. Informally, if a perturbed algorithm is obtained by replacing a quantity with a random plug-in estimator, it is possible to write down a pseudomarginal algorithm based on the same plug-in estimator. This pseudomarginal algorithm will have a simple formula for its associated stationary measure, and this stationary measure can be analyzed directly. See \cite{Quiroz19} for a prototypical example of this approach with useful estimates. A distinctive approach is found in \cite{quiroz2021block}, where the draws from the perturbed posterior are reweighted such that expectations with respect to the unperturbed target can be consistently estimated, avoiding the need of analysing the stationary perturbed measure.

This leads to a natural question: if one can \textit{typically} tweak a perturbed algorithm into a similar algorithm for which the stationary distribution is available, or if the samples can be reweighted as in \cite{quiroz2021block}, why would anyone use a perturbation analysis along the lines of Theorem \ref{1th:Z_m}? 

We believe that there are two main answers. The obvious one is that these tweaks may not be available. For example, at the time of writing there are no known tweaks for approaches such as \cite{Chen22} that are based on core sets. The second answer is that these tweaks can be computationally expensive when they are available. As discussed in the detailed theoretical analysis \cite{andrieu2022comparison}, pseudomarginal algorithms can be much slower than their ``base" algorithms and often even lack spectral gaps. These poor theoretical properties can also be evident in applications. Pseudomarginal methods can be made more efficient by correlating the estimates in the numerator and denominator as in \cite{deligiannidis2018correlated, tran2016block}. The correlated pseudomarginal method has been successfully implemented in many applications, e.g. \cite{friedli2023inference, yang2022, golightly2019correlated, wiqvist2021efficient}. However, using a deterministic approximation of the likelihood, albeit perturbed, may sometimes be more efficient as shown in \cite{friedli2023solving}. For these reasons, it remains common to run perturbed Markov chains for which the associated stationary distribution has no known simple formula.

\subsection{Basic Principles: What Is Theory For?}

It is natural to ask how theoretical results are relevant to practice. The best-case result would be to get upper bounds on the error (as in Theorem \ref{1th:Z_m}) and lower bounds that are nearly equal. We could then use these matching bounds to do hyperparameter tuning. Unfortunately, this is almost never possible for realistic Markov chains - whether in perturbation theory or more generally \cite{Honest01}. While this question is too broad 
%to answer here, we give a few partial answers:
to address here comprehensively, we give a few partial answers:

\begin{enumerate}
\item \textbf{Heuristic Tuning:} Even when error bounds don't match, it is common to optimize an upper (or lower) bound as a proxy for the true error. This tuning is rarely optimal, but can be a good starting point. Sometimes, theory is available to help \cite{MEDINAAGUAYO20202200}.
\item \textbf{Basic Confirmation:} When writing classes of MCMC algorithms that may be used by a wide audience, it is not possible to check that basic desirable properties (such as geometric ergodicity or consistency) hold by exhaustive simulation. Theoretical bounds such as Theorem \ref{1th:Z_m} give a sort of ``sanity check" by establishing that, at the very least, all algorithms in a class satisfy such basic properties. They are also often seen as weak evidence that good performance in a limited number of test cases is likely to be seen across many algorithms in the class.
 \item \textbf{Mapping the Territory:} A priori, it is often not clear what approaches to perturbed Markov chains have \textit{any} chance of producing usable algorithms. Theoretical results often provide a clear picture of how certain quantities (e.g. pointwise error) must scale with other quantities (e.g. dimension). This is our main point of view in Section \ref{SecExampleAppsGeneric}. This point of view can be quite useful in finding a good algorithm, but it occurs \textit{earlier} in the process than hyperparameter tuning. To continue our running example of subsample-based MCMC: hyperparameter tuning will tell you the best rate at which to subsample, and you can usually do this after implementing the algorithm. Basic theory will tell you that no amount of subsampling is enough without control variates, and will suggest control variates that are statistically efficient. These sorts of considerations usually happen before implementation and before hyperparameter tuning.
\end{enumerate}

\section{A Survey of Mathematical Results} \label{SecMathPert}

\subsection{A Brief History}

Perturbation theory has a long history in the analysis of matrices and linear operators, as summarized in the classic textbook \cite{Kato1995}. Of course, making stronger assumptions on the operators could give stronger perturbation bounds. Some early specializations of perturbation theory to statistics came from the Davis-Kahan theorem and extensions \cite{davis1970rotation,SamDav15}, which analyzed perturbations in $L^{2}$. While these general results have been used to analyze operators that can be viewed as transition kernels of Markov chains \cite{Lux07}, these results have not been the main approach to proving perturbation bounds for MCMC analysis. The first applications of generic perturbation bounds to Markov chains came in \cite{kart86,Seneta91}. At around the same time as these first generic perturbation bounds, similar ideas were being specialized to Markov chains appearing in queuing theory \cite{ho83,CAO1986433} as well as to Monte Carlo methods for stochastic differential equations \cite{TalayNumericSDE90}. This perturbation-based approach remains popular in the queuing theory literature (see \cite{Araman2022} and the references therein).

\subsubsection{Applications to State-Space Approximations}

The first application to the analysis of MCMC that we are aware of was \cite{BREYER2001123}, which showed that ``rounding off" numbers to machine precision introduced negligible total bias to the results of Monte Carlo calculations. However, for nearly thirty years, perturbation analysis was popular in queuing theory and other areas of applied probability, but was rarely used in the analysis of MCMC algorithms. This was changed by two major developments amongst scientists: practitioners in all areas found that popular MCMC algorithms were becoming less tractable in the face of very large datasets and complex models, and computer scientists became very interested in understanding stochastic gradient descent and other gradient-based Markov chains. This resulted 
in at least two large threads of research on perturbations in the context of MCMC: perturbations based on approximating likelihoods and perturbations based on approximating continuous-time algorithms.

Practitioners working on problems for which typical MCMC algorithms were intractable developed cheap approximations. The simplest of these were essentially ``plug-in" approximations, replacing an expensive calculation (such as \eqref{EqProdPost}) by a tractable approximation (such as \eqref{EqSimpleProxy}). In a few cases it was possible to develop auxiliary-variable tricks to render these plug-in algorithms exact (see e.g. \cite{pseudomarginal,fly15,sherlock2017pseudo}), but typically one would analyze these ``plug-in" algorithms by viewing them as perturbations of exact algorithms. This lead to a long string of papers proposing new approximate algorithms in a range of areas, including subsampling \cite{Korattikara2013AusterityIM, bardenet2014towards,pmlr-v80-desa18a,Quiroz19} especially in the case of gradient-based methods \cite{wellingSGLD11,pmlr-v32-cheni14,SHMCMCSurvey, dang2019}, expensive forward models \cite{ApproxBayes10,DodSub15,Conrad16,Conrad2016ParallelLA,stuart2018posterior,yan2019adaptive}, and ``doubly-intractable" distributions \cite{MurraryExchange06,kang2023measuring,ERGM23,Roulette15, yang2022,habeck2020stability}. Alongside the new algorithms came papers developing  basic theory \cite{mitrophanov2005sensitivity,Ferre_ledoux_2013,johndrow2017error,alquier2016noisy,medina2016stability,DALALYAN20195278,SchweizerPerturb18,MEDINAAGUAYO20202200,negrea_rosenthal_2021} and others using it to give practical advice on how to make the algorithms perform well \cite{NIPS2016_9b698eb3,chatterji2018theory,CVSGMCMC19,NEURIPS2019_c3535feb,Quiroz18} and to draw conclusions about when popular algorithms will perform poorly \cite{betancourt15,bardTall17,Nagapetyan2017,NEURIPS2018_335cd1b9,johndrow2020free}. See Section \ref{SecExampleAppsGeneric} for a discussion of the practical lessons one can learn from these results.

\subsubsection{Applications to Time Discretizations}

In an essentially-unrelated collection of research papers, theorists began applying perturbation-like results to analyze algorithms based on discrete-time approximation to continuous-time algorithms. The simplest and most famous of these algorithms is the Langevin algorithm, which we now recall. Consider a target distribution with density $e^{U(x)}$ on $\mathbb{R}^{d}$. The associated Langevin dynamics are the Ito diffusion 
\begin{equation} \label{EqLangevin}
    \dot{X} = \nabla U(X) + \sqrt{2} \dot{W},
\end{equation}
where $W$ is standard Brownian noise. This diffusion has $e^{U(x)}$ as its stationary measure, but of course sampling from the diffusion is not typically tractable. The natural Euler approximation to equation \eqref{EqLangevin} with time-step of length $\delta > 0$ is:
\begin{equation} \label{EqULA}
X_{t+1} = X_{t} + \delta \nabla U(X_{t}) + \sqrt{2 \delta} \zeta_{k},
\end{equation}
where $\zeta_{1},\zeta_{2},\ldots$ are iid standard $d$-dimensional Gaussian random variables. The dynamics in equation \eqref{EqULA} are called the \textit{unadjusted Langevin algorithm} (ULA)\footnote{It is common to add a Metropolis-Hastings correction step after applying the recursion in equation \eqref{EqULA}. The resulting algorithm, called the Metropolis-adjusted Langevin algorithm (MALA), has the original target $e^{U(x)}$ as its stationary measure. As discussed in \cite{Durmus2016HighdimensionalBI}, the faster-but-biased ULA algorithm is sometimes preferred to the slower-but-unbiased MALA algorithm.}.

There are many bounds on the error of ULA and other algorithms based on discretizations of continuous processes. Some prominent bounds on the error of ULA are the sequence of papers \cite{TalayNumericSDE90,Durmus2016HighdimensionalBI,Durmus17,DALALYAN20195278}, all of which seem to critically involve bounding the growth rate of the difference $\|X_{t} - X(\delta t)\|$ between the ULA dynamics  $X_{t}$ from equation \eqref{EqULA} and the Langevin dynamics $X(t)$ from equation \eqref{EqLangevin}. In other words, they view the dynamics of the ULA algorithm as a small perturbation of the Langevin dynamics. 

One critical ingredient in obtaining new perturbation-based bounds, which is on display in recent analyses of ULA, is to use the perturbation to bound ``less" than the full path. Very roughly, the idea is to use a perturbation bound to show that the perturbed chain itself inherits certain functional inequalities from the original chain, rather than showing that paths of the perturbed chain are in some sense close to paths of the original chain. Recall that the proof of the toy result Theorem \ref{1th:Z_m} proceeds by viewing entire paths of two Markov chains as being close. In contrast, the analysis of ULA in \cite{Bou23} views the perturbation more abstractly: it is used only to show that ULA inherits a contraction bound from the Langevin dynamics, not to bound the distance between long Markov chain paths (see Theorem 3.1). Although  \cite{Bou23} does not involve bounding distances between entire paths, the perturbation is still somewhat concrete, as it involves bounding distances between individual full steps of the algorithm.  In \cite{NEURIPS2019_65a99bb7}, it is viewed even more abstractly than this: it is used to show that ULA inherits a bound based on log-Sobolev inequalities from the Langevin dynamics (see Lemma 3). 

Although we have focused on the analysis of the non-asymptotic error of ULA, perturbation bounds are popular in other analyses of other discrete approximations to continuous algorithms. There is a similar sequence of papers analyzing the Hamiltonian Monte Carlo (HMC) algorithm introduced in \cite{Neal95}, which can be viewed as a multi-step analogue to the ULA or MALA algorithm. As in the analysis of ULA, early papers viewed the paths of standard HMC algorithms implemented with Euler or Verlet integrators as small perturbations of paths of a natural continuous HMC algorithm \cite{pmlr-v89-mangoubi19a,mangoubi2018dimensionally}, while later papers used perturbation bounds to obtain bounds on contraction rates \cite{BouEbZim20} or even more abstract quantities \cite{Chen20}. Similar progress has occurred in the analysis of other gradient-based methods; see e.g. \cite{zhang2023improved}. Recent work has taken advantage of a long line of papers from this viewpoint in the optimization literature, such as \cite{NEURIPS2019_a9986cb0}, to propose and analyze new MCMC methods \cite{Nest21}.

\subsubsection{Towards Practical Advice}

All of this mathematical discussion has avoided the key question: what practical advice can one obtain from these convergence estimates? Unfortunately, it is very difficult to use these bounds to directly compare algorithms. The basic issue is that no direct comparisons, such as Peskun ordering and related bounds (see e.g. the survey \cite{MiraOrderingSurvey01} or related methods \cite{CompExactApprox16,andrieu2022comparison}), are available for these sorts of algorithms. So far it is only possible to conclude that one algorithm is better than another by obtaining \textit{upper} bounds on the error of the first algorithm that are smaller than \textit{lower} bounds on the error of the second algorithm. Bounds that apply very broadly, such as \cite{BouEbZim20}, are likely to be quantitatively loose, and so such a comparison is nearly impossible. Of course, the fact that it is very difficult to compare MCMC algorithms by obtaining sharp convergence estimates is a well-known problem in MCMC theory \cite{jones2001honest}.

Until recently, non-asymptotic error bounds have largely been used as rough confirmation for heuristics obtained in other ways. For example, scaling limit analyses suggests that HMC would perform better than MALA on high-dimensional problems \cite{LangevinScaling12,ScalingHMC13}, and these rates are consistent with the literature on non-asymptotic error of HMC and MALA. It is also very common to compare only \textit{upper bounds} on convergence rates of two algorithms (as in the informal discussion in Section 2.4 of \cite{SHMCMCSurvey}) or only \textit{lower bounds} on convergence rates of two algorithms (as in the main conclusion of \cite{Nagapetyan2017}). As most authors are quick to point out, these sorts of comparisons are suggestive but do not allow for strong conclusions.

We are aware of only one paper that obtains rigorous comparisons of these sorts of gradient-based methods: \cite{chen2023does} obtains nearly-matching upper and lower bounds on the convergence rates of MALA and HMC under fairly strong conditions on the target distributions. These allow one to conclude a heuristic that has been well-known among practitioners for many years: HMC performs better than MALA for high-dimensional problems.

 Although the motivations of these two threads of research papers are quite different, some important connections have emerged. Some of the most interesting connections are related to stochastic gradient methods, as surveyed in \cite{CompleteRecipe15,SHMCMCSurvey}. For any algorithm that involves computing a gradient (such as MALA or HMC), it is straightforward to write down a subsampling analogue of the algorithm by replacing any appearances of the gradient by a subsample-based estimate. It is natural to analyze these algorithms as ``double perturbations" of the original diffusion, but in fact there are sometimes more effective options. For example, the stochastic gradient HMC algorithm of \cite{pmlr-v32-cheni14} is best viewed as a perturbation of the solution to the following ``underdamped" version of the Langevin diffusion \eqref{EqLangevin}:

\begin{equation*}
\dot{Z} = -\tau \, Z \, dt - \nabla U(X) dt + \sqrt{2 \tau \beta^{-1}} \dot{W}, \quad \dot{X} = Z dt,
\end{equation*}
rather than viewing it as a straightforward perturbation for the original HMC algorithm. Recognizing these alternative viewpoints leads to tighter analyses - see e.g. \cite{JMLR:v17:15-494}.

\subsection{Current State-of-the-Art}

We turn to the current state-of-the-art regarding the theory of perturbations of Markov chains.
The objects of interest are the unperturbed Markov chain
$(X_n)_{n\in\mathbb{N}_0}$ with transition kernel $Q$ and the perturbed Markov chain $(\widetilde{X}_n)_{n\in\mathbb{N}_0}$ with transition kernel $K$ on $(\Omega,\mathcal{F})$. Recall that for convenience we denote, for any $n\in\mathbb{N}_0$, by $p_n$ the distribution of $X_n$ and by $\widetilde{p}_n$ the distribution of $\widetilde{X}_n$.

Now we consider the two main threads of research regarding the aforementioned perturbation theory. The ``classical" approach involves arguments that are reminiscent of the proof of Theorem \ref{1th:Z_m}: the core of the argument is to start with a ``telescoping sum" decomposition along the lines of \cite[equation (3.3)]{SchweizerPerturb18} of the form
\begin{equation} \label{IneqTeleSum}
% \nu K^{t} - \mu Q^{t} = (\nu - \mu) Q^{t} + \sum_{s=0^{t-1}} \nu K^{s}(K-Q)Q^{t-s-1}  
\widetilde{p}_n  -p_n
= \widetilde{p}_0 K^n -p_0 Q^n  
= (\widetilde{p}_0 - p_0) Q^n + \sum_{i=0}^{n-1} \widetilde{p_0} K^i (K-Q)Q^{n-i-1}
\end{equation}
in combination with the triangle inequality. Such arguments give bounds of the difference of $Q^{n}(x,\cdot)$ and $K^{n}(x,\cdot)$ in some distance, and they 
allow, under suitable assumptions, to
derive bounds on 
% $(\mu-\nu)$ 
the corresponding stationary distribution.
% as consequences of the bounds on $(Q^{n}(x,\cdot) - K^{n}(x,\cdot))$. 
In contrast to that 
the ``abstract" approach is focused on showing that $K$ inherits certain ``good" convergence properties from $Q$, then analyzing $K$ in terms of these `good' properties.

\subsubsection{The Classical Approach} \label{SecClassicalApproach}

The classical thread is dominant in the early probability literature and in the statistics literature on MCMC \cite{kart86,Seneta91,mitrophanov2005sensitivity,Ferre_ledoux_2013,PertStat13,Kartashov2013MaximalCP,CASWELL20131727,johndrow2017error,SchweizerPerturb18,negrea_rosenthal_2021}. The main improvements have come from investigating how to suitably measure the difference between perturbed and unperturbed transition kernel. There are at least two useful ingredients:

\begin{enumerate}
    \item \textbf{Wasserstein distance and other norms:} Early papers such as \cite{mitrophanov2005sensitivity,kart86} use strong quantities, such as the total variation distance or so called $V$-norms, to measure perturbations.
    Recent papers rely on the Wasserstein distance \cite{SchweizerPerturb18,hosseini2023spectral} that on the one hand also allows to treat the total variation and $V$-norm setting \cite{johndrow2017error} as special cases  and on the other hand provides flexibility in taking a metric fine structure of the underlying space into account. However,
    a notable exception is \cite{negrea_rosenthal_2021} that relies on the $L^{2}$ distance and follows an $L^2$-operator norm approach in the light of perturbation theory of linear operators on Hilbert spaces \cite{davis1970rotation}.
    \item \textbf{Drift or Lyapunov condition:} 
    It turned out to be particularly fruitful to 
    formulate a \textit{drift or Lyapunov condition} 
    regarding the perturbed transition kernel, that is, assume that there is a function $V \, : \, \Omega \to [1,\infty)$ and constants $0<\beta<1$, $0 \leq L < \infty$ such that 
    \begin{equation} \label{IneqLyapAssumption}
    (KV)(x) \leq (1-\beta) V(x) + L.
    \end{equation}
    Such a condition can always be trivially satisfied with $V\equiv 1$. In this sense the requirement is not restrictive and one may conclude that a non-trivial $V$ only adds some (maybe) useful knowledge. The main idea \cite{SchweizerPerturb18,MEDINAAGUAYO20202200,hosseini2023spectral} now is to use $V$ as weight function in the distance or norm
    that one considers to measure the difference of $K$ and $Q$. This ansatz also allows to treat  proxies on truncated state spaces \cite[Theorem~9]{MEDINAAGUAYO20202200}.
\end{enumerate}

In Example~\ref{ExEucDist} and Example~\ref{ExLyap} 
we saw already advantages of using the Wasserstein distance. In particular, the approximation error between $Q$ and $K$ in Example~\ref{ExEucDist} would have been constant if measured in the total variation norm such as in  \cite{mitrophanov2005sensitivity}. However, a proper treatment of Example~\ref{ExLyap} is still missing. We turn to this now and provide without proof a slightly simplified version of \cite[Theorem~3.1]{SchweizerPerturb18} to illustrate the benefits of the aforementioned ingredients also within this example.
\begin{theorem} \label{2pertbnd}
Assume that $Q$ satisfies the contraction property  \eqref{IneqSimpleCont} and assume that $K$ with $V\colon \Omega \to [1,\infty)$ satisfies the Lyapunov condition \eqref{IneqLyapAssumption}. Let
\[
 \varepsilon := \sup_{x\in\Omega} \frac{W_d(K(x,\cdot),Q(x,\cdot))}{V(x)}
 \qquad \text{and} \qquad
 \kappa:= \max\left\{\frac{L}{\beta}, \int_\Omega V(x) \,\widetilde p_0({\rm d}x) \right\}.
\]
Then
\[
W_d(p_0 Q^n,\widetilde{p}_0K^n) \leq (1-\alpha)^n W_d(p_0,\widetilde{p}_0) + \frac{\varepsilon \kappa}{\alpha}.
\]
Moreover, if $Q,K$ have stationary measures $\pi_Q,\pi_K$ with $W_d(\pi_Q,\pi_K)<\infty$, then 
\begin{equation} 
W_{d}(\pi_Q,\pi_K) \leq  \frac{ \varepsilon}{\alpha} \frac{L}{\beta}.
\end{equation}
\end{theorem}

We want to point out that $\varepsilon$ quantitatively measures the difference between perturbed and unperturbed kernel by incorporating the function $V$ in the denominator. This allows us to treat cases where $W_d(K(x,\cdot),Q(x,\cdot))$ is not uniformly bounded in $x$. For illustrating the former theorem we reconsider Example~\ref{ExLyap}.

\begin{example}[Restriction of  Theorem~\ref{1th:Z_m} (continued)] 

We verify that the perturbed Markov chain given by $\widetilde{X}_{n+1} = (1-\widetilde \alpha) \widetilde X_n + Z_n$ satisfies \eqref{IneqLyapAssumption} for $x\mapsto V(x) := 1 + \vert x\vert$ and apply Theorem~\ref{2pertbnd}. We have
\begin{align*}
(KV)(x) =  1+ \int_{\mathbb{R}} \vert (1-\widetilde{\alpha}) x+z \vert \mu({\rm d}z)
\leq ( 1-\widetilde{\alpha}) V(x)+ \widetilde{\alpha} + \int_{\mathbb{R}} \vert z \vert \mu({\rm d}z),
\end{align*}
such that with $\beta:= \widetilde{\alpha}$ and finite $L:= \widetilde{\alpha} + \int_{\mathbb{R}} \vert z \vert \mu({\rm d}z)$ condition \eqref{IneqLyapAssumption} holds.
Then, combining \eqref{upper_err_bnd_restr_ex} and \eqref{eq:lower_err_bnd_restr_ex} yields
\[
\varepsilon = \vert \alpha-\widetilde{\alpha}\vert,
\]
which indicates that for $\alpha$ being close to $\widetilde{\alpha}$ the difference of the unperturbed and perturbed transition kernel is small. Consequently, for $p_0=\delta_x$, $\widetilde{p}_0 = \delta_y$ with $x,y\in \Omega$ Theorem~\ref{2pertbnd}
implies
\[
W_d(Q^n(x,\cdot),K^n(y,\cdot))
\leq (1-\alpha)^n \vert x-y\vert + \frac{\vert \alpha-\widetilde{\alpha}\vert}{\alpha} 
\max\left\{\frac{
\int_{\mathbb{R}} \vert z \vert \mu({\rm d}z)}{\widetilde{\alpha}}+1,\vert y \vert +1
\right\}.
\]
Moreover, for the stationary measures $\pi_Q$ and $\pi_K$ holds
\[
W_d(\pi_Q,\pi_K) \leq 
\frac{\vert \alpha-\widetilde{\alpha}\vert}{\alpha} \left(\frac{\int_{\mathbb{R}} \vert z \vert \mu({\rm d}z)}{\widetilde{\alpha}}+1 \right).
\]    
\end{example}

Intuitively, from the quantity $\varepsilon$ in the previous theorem we can conclude that the Lyapunov condition of the perturbed transition kernel allows us to identify the essential part of the state space.  Namely, in that area where the Lyapunov function is large, a more loose proxy in terms of $W_d(Q(x,\cdot),K(x,\cdot))$ is sufficient. We want to point out that the same Wasserstein distance appears in the contraction property \eqref{IneqSimpleCont} and in $\varepsilon$ of Theorem~\ref{2pertbnd}. This may be disadvantageous, for a discussion see \cite{Ferre_ledoux_2013,SchweizerPerturb18}. Regarding this issue, we mention that \cite{Ferre_ledoux_2013,SchweizerPerturb18,MEDINAAGUAYO20202200}
also provide theoretical bounds where the contraction is measured in a $V$-norm and the proximality in terms of the total variation distance.

Now we illustrate the beneficial application of Theorem~\ref{2pertbnd} within a classical Metropolis algorithm setting. 

\begin{example}[Perturbed independent Metropolis] \label{ExIMH}
    Let $\Omega=\mathbb{R}^d$ and let the Lebesgue probability density function (pdf) of the distribution of interest be $\pi$. Suppose $\pi \propto e^{U}$ for some appropriate log-density $U\colon \mathbb{R}^d \to \mathbb{R}$.
    % $\rho\colon \mathbb{R^d} \to \mathbb{R}_+$. 
    By $\mu$ we denote the Lebesgue pdf of the proposal distribution that is independent of the current state. Then, the Markov chain based on the independent Metropolis algorithm with proposal $\mu$ takes the role of the unperturbed, ideal process. 
    With acceptance probability $a(x,y) = \min\{1,\frac{e^{U(y)}}{\mu(y)} \frac{\mu(x)}{e^{U(x)}}\}$ for $x,y\in \Omega$ the corresponding transition kernel is given by
    \[
    Q(x,A) = \int_A a(x,y) \mu(y) {\rm d }y + \mathbf{1}_A(x) \left( 1-\int_{\mathbb{R}^d} a(x,y) \mu(y) {\rm d }y \right),\quad A\in \mathcal{B}(\mathbb{R}^d). 
    \]
    We assume that the proposal and target distribution satisfy $\mu(x)/\pi(x)\geq \alpha>0$ for all $x\in \Omega$, so that by virtue of the ``smallness" derived in the proof of \cite[Theorem~2.1]{mengersen1996rates} and \cite[Lemma~18.2.7]{douc2018markov} we have
    \begin{equation}
    \label{eq:contr_IMH}
    d_{TV}(Q(x,\cdot),Q(y,\cdot)) \leq 1-\alpha, \qquad \forall x,y\in \Omega,\; x\neq y. 
    \end{equation}
    Therefore, the contraction property \eqref{IneqSimpleCont} is satisfied for $d_{TV}$. 
    
    We consider the setting where on some part of the state space, say $\widetilde{\Omega}\subseteq \Omega$, only the computation of a noisy evaluation of $U$ is available. For a real-valued random variable $\xi_{x}\sim \mathcal{N}(0,\sigma^2)$ with noise variance $\sigma^2>0$, we assume we are able to evaluate
    $\widehat{U}(x) = (U(x)+\xi_{x}) \mathbf{1}_{x\in\widetilde{\Omega}} + U(x)\mathbf{1}_{x\not\in\widetilde{\Omega}} $, i.e., 
    $\xi_{x}$ models the noise for the evaluation on $\widetilde{\Omega}$. 
    % \[
    % \widehat{U}(x) 
    % = \begin{cases}
    %     U(x)+\xi_{x} & x\in\widetilde{\Omega}\\
    %     U(x) & x\in \widetilde{\Omega}^c.
    %   \end{cases}      
    % \]
    Now a transition from $x\in \Omega$ to $y\in \Omega$ of the perturbed independent Metropolis algorithm that we consider works as follows:
    \begin{enumerate}
        \item Propose $Z\sim \mu$, call the result $z\in \Omega$;
        \item Compute $\widehat{U}(x)$ and $\widehat{U}(z)$ independently;
        \item With probability
        $
            \widehat a(x,z) := \min\left\{1, \frac{e^{\widehat{U}(z)}}{\mu(z)} \frac{\mu(x)}{e^{\widehat{U}(x)}}\right\},
        $
        accept $z$, i.e., return $y:=z$, and with probability $1-\widehat{a}(x,z)$ reject, i.e., return $y:=x$.
    \end{enumerate}
    The corresponding transition kernel is given by
    \[
    K(x,A) = \int_A 
    \mathbb{E}[\widehat{a}(x,y)] \mu(y) {\rm d }y + \mathbf{1}_A(x) \left( 1-\int_{\mathbb{R}^d} \mathbb{E}[\widehat{a}(x,y)] \mu(y) {\rm d }y \right),\quad A\in \mathcal{B}(\mathbb{R}^d). 
    \]
    A straightforward calculation reveals (cf. \cite[Lemma~4.1]{SchweizerPerturb18} or \cite[Proposition~12]{MEDINAAGUAYO20202200}) that
    \begin{equation}
    \label{eq:d_tv_diff_acc_prob}
    d_{TV}(Q(x,\cdot),K(x,\cdot)) 
    \leq 2 \int_{\Omega} \vert a(x,y)-\mathbb{E}[\widehat{a}(x,y)] \vert \mu(y) {\rm d}y.
    \end{equation}
    This indicates that the difference of the (expected) acceptance probabilities is essential.
    \begin{lemma}
        For any $x,y\in \Omega$ we have 
        \[
            \vert a(x,y) - \mathbb{E}[\widehat{a}(x,y)] \vert \leq  
            \begin{cases}
                0, & x,y\not\in \widetilde{\Omega},\\
                \sigma^2 e^{\sigma^2},& x\in\widetilde{\Omega}\;\;\text{or}\;\; y\in\widetilde{\Omega}.
            \end{cases}
        \]
    \end{lemma}
    \begin{proof}
        Set $r(x,y):= \frac{e^{U(y)}}{\mu(y)} \frac{\mu(x)}{e^{U(x)}}$ and let $\xi_{x},\eta\sim \mathcal{N}(0,\sigma^2)$ be independent normally distributed real-valued random variables. 
        
        On the one hand, 
        by the concavity of $t\mapsto \min\{1,t\}$ with Jensen's inequality we have
        \begin{align*}
            \mathbb{E}[\widehat{a}(x,y)] \leq \min\left\{ 
            1,r(x,y)\mathbb{E}[e^{\xi_{x}\mathbf{1}_{y\in\widetilde{\Omega}}-\eta \mathbf{1}_{x\in\widetilde{\Omega}}}]\right\}.
        \end{align*}
        Observe that for $x,y\not\in \widetilde{\Omega}$ holds
        $ \mathbb{E}[\widehat{a}(x,y)] \leq \alpha(x,y)$. 
        In the case $x\in\widetilde{\Omega}$ or $y\in\widetilde{\Omega}$, 
        by the independence of the random variables, the expectations inside the minimum are in each subcase taken w.r.t. a lognormal distribution, such that $ \mathbb{E}[\widehat{a}(x,y)] \leq \min\{1,r(x,y) e^{\sigma^2}\} \leq a(x,y) e^{\sigma^2}.$ Here the last inequality follows by the fact that $\min\{1,c_1 c_2\} \leq \min\{1,c_1\} c_2$ for any $c_1\geq0$ and $c_2\geq 1$.

        On the other hand, by exploiting $\min\{1,c_1 c_2\} \geq \min\{1,c_1\} \min\{1,c_2\}$ for any $c_1,c_2\geq0$ we have
        \[
            \mathbb{E}[\widehat{a}(x,y)] \geq a(x,y) \mathbb{E}[\min\{1,e^{\xi_{x}\mathbf{1}_{y\in\widetilde{\Omega}}-\eta \mathbf{1}_{x\in\widetilde{\Omega}}}\}]
            \geq a(x,y).
        \]
        Here the last inequality follows by the fact that $\mathbb{E}[\min\{1,e^Z\}] = 1/2(1+e^{v^2/2})\geq 1$ for any $Z\sim\mathcal{N}(0,v^2)$ with $v^2>0$. 

        Therefore, taking both inequalities together yields
        \begin{align*}
            \vert a(x,y) - \mathbb{E}[\widehat{a}(x,y)] \vert & \leq \mathbb{E}[\widehat{a}(x,y)] - a(x,y)  \leq 0\cdot \mathbf{1}_{x,y\not\in\widetilde{\Omega}} + (e^{\sigma^2}-1)a(x,y)\mathbf{1}_{x\in\widetilde{\Omega}
            % \;\text{\rm or}\;
            \lor
            y\in\widetilde{\Omega}}\\
            & \leq 
            % \mathbf{1}_{x,y\not\in\widetilde{\Omega}} + 
            \sigma^2 e^{\sigma^2}\mathbf{1}_{x\in\widetilde{\Omega}
            % \;\text{\rm or}\;
            \lor
            y\in\widetilde{\Omega}},
        \end{align*}
    which finishes the proof.
    \end{proof}

    By the former lemma and \eqref{eq:d_tv_diff_acc_prob} we have
    \begin{equation}
    \label{eq:ex_proxy}
     d_{TV}(Q(x,\cdot),K(x,\cdot)) \leq 2 \sigma^2 e^{\sigma^2} \left( \mathbf{1}_{x\in\widetilde{\Omega}}+\mu(\widetilde{\Omega}) \mathbf{1}_{x\not\in \widetilde{\Omega}} \right)
     \leq 2\sigma^2 e^{\sigma^2},
    \end{equation}
    such that a direct application of Theorem~\ref{2pertbnd} with $V\equiv 1$,
    $\varepsilon\leq 2\sigma^2 e^{\sigma^2}$, $\kappa=1$ and initial distributions, for example $p_0=\widetilde{p}_0=\mu$, gives
    \begin{equation}
    \label{eq:IMH_1st_est}
    d_{TV}(\mu Q^n,\mu K^n) \leq \frac{2\sigma^2 e^{\sigma^2}}{\alpha}.
    \end{equation}
    This is ``nice" in the sense that, if the variance $\sigma^2$ of the noise goes to zero, then the difference of the distributions of the $n$th step of the unperturbed and perturbed Markov chain also goes to zero.
    However, at this point the estimate cannot capture that the noise only appears on $\widetilde{\Omega}$. We see how a non-trivial Lyapunov function can help.
    \begin{lemma}
        Recall that we assume $\mu(x)/\pi(x) \geq \alpha>0$ for all $x\in\Omega$. Let $V(x) = 1+ R \mathbf{1}_{x\in \widetilde{\Omega}}$ for arbitrary $R\geq 0$. Then
        \[
        KV(x) \leq (1-\alpha) V(x) + \alpha +R \mu(\widetilde{\Omega}).
        \]
    \end{lemma}
    \begin{proof}
        By the fact that $\alpha$ is a lower bound of a ratio of pdfs we have $\alpha\leq 1$.
        Moreover, by $a(x,y) \leq \mathbb{E}[\widehat{a}(x,y)] \leq 1$ we have
        \begin{align*}
            K(x,\widetilde{\Omega
            }) \leq \mu(\widetilde{\Omega}) + \left(1-\int_{\mathbb{R}^d} 
            a(x,y) \mu(y) {\rm d}y\right)   \mathbf{1}_{x\in\widetilde{\Omega}} 
            \leq \mu(\widetilde{\Omega}) + \left(1-\alpha\right)   \mathbf{1}_{x\in\widetilde{\Omega}},
        \end{align*}
        where we used  $\int_{\mathbb{R}^d} \min\big\{1,\frac{\pi(y)}{\mu(y)} \frac{\mu(x)}{
        \pi(x)} \big\}\mu(y) {\rm d }y = \int_{\mathbb{R}^d} \min\big\{\frac{\mu(y)}{\pi(y)}, \frac{\mu(x)}{
        \pi(x)} \big\}\pi(y) {\rm d }y \geq \alpha$. With this at hand we deduce
        \[
        KV(x) = (1-\alpha)[1+R\mathbf{1}_{x\in\widetilde{\Omega}}] + \alpha + R[K(x,\widetilde{\Omega})-(1-\alpha)\mathbf{1}_{x\in\widetilde{\Omega}}]
        \leq (1-\alpha) V(x) + \alpha + R \mu(\widetilde{\Omega})
        \]
        and the proof is completed.
    \end{proof}
    We aim to apply Theorem~\ref{2pertbnd} by employing the Lyapunov function
    $V(x) = 1 + \mu(\widetilde{\Omega})^{-1} \mathbf{1}_{x\in \widetilde{\Omega}}$, i.e., we set $R=\mu(\widetilde{\Omega})^{-1}$. As before let $p_0=\widetilde{p}_0 = \mu$, so that $\kappa = 1+\alpha^{-1}$. By exploiting the middle estimation in \eqref{eq:ex_proxy} we obtain
    \begin{align*}
    \varepsilon & = \max\Big\{ \sup_{x\in \widetilde{\Omega}}\frac{d_{TV}(Q(x,\cdot),K(x,\cdot))}{V(x)},\sup_{x\in \Omega\setminus\widetilde{\Omega}}\frac{d_{TV}(Q(x,\cdot),K(x,\cdot))}{V(x)}\Big\} \\
    &\leq 2\sigma^2 e^{\sigma^2} 
    \max\{(1+\mu(\widetilde{\Omega})^{-1})^{-1},\mu(\widetilde{\Omega})\}
    = 2 \sigma^2 e^{\sigma^2} \mu(\widetilde{\Omega}).
    \end{align*}
    Therefore, we have
    \[
    d_{TV}(\mu Q^n,\mu K^n) \leq \frac{2 \sigma^2 e^{\sigma^2} \mu(\widetilde{\Omega})}{\alpha} \Big(1+\frac{1}{\alpha}\Big). 
    \]
    The advantage of the former estimate compared to \eqref{eq:IMH_1st_est} lies in the fact that the ``size" of $\widetilde \Omega$ appears. Namely, if the part of the state space $\widetilde\Omega$, where only noisy evaluations of the log-likelihood are available, is ``small" in the sense that $\mu(\widetilde\Omega)$ is ``small", then the difference between the $n$th step distributions of perturbed and unperturbed Markov chain is also ``small".   
\end{example}

\subsubsection{The Abstract Approach}

The ``abstract" or ``inherited" approach uses perturbation bounds to show that the perturbed Markov chain with transition kernel  $K$ inherits some ``good" property from the unperturbed chain with kernel $Q$. We give a simple scenario from the classical literature. 
Suppose we have \eqref{IneqSimpleCont} and \eqref{IneqSimpleApprox} with respect to the total variation distance, that is, $Q$ satisfies
% if $K$ exhibits the TV contraction property 
\begin{equation} \label{IneqContTV}
\sup_{x,y \in \Omega} d_{TV}(Q(x,\cdot), Q(y,\cdot)) \leq 1-\alpha
\end{equation}
for some $\alpha>0$ and
$K$ is a ``small" perturbation of $Q$ in the sense
\begin{equation}    
\label{IneqAppTV}
\sup_{x \in \Omega} d_{TV}(K(x,\cdot),Q(x,\cdot)) \leq \varepsilon,
\end{equation}
with $2\varepsilon <\alpha$.
Then, the triangle inequality gives the contraction of $K$ by
\begin{equation*}
\sup_{x,y \in \Omega} d_{TV}(K(x,\cdot), K(y,\cdot)) \leq 1-\alpha+2\varepsilon<1,
\end{equation*}
which implies the existence of a stationary distribution $\pi_K$ of $K$ and geometric convergence of $K^n(x,\cdot)$ to $\pi_K$ in the total variation distance. 

We illustrate the benefits of this simple estimate in combination with Theorem~\ref{2pertbnd} in the independent Metropolis setting.

\begin{example}[Perturbed independent Metropolis (continued)] \label{ExIMH_inheritance}
 We consider the same setting and assumptions as in Example~\ref{ExIMH}. Note that \eqref{IneqContTV} is satisfied by \eqref{eq:contr_IMH} and by \eqref{eq:ex_proxy} with
$\varepsilon := 2\sigma^2e^{\sigma^2}$ the condition \eqref{IneqAppTV} holds for $0\leq \sigma^2 <\alpha/12$. 

This gives the existence of a stationary distribution
$\pi_K$ of the perturbed independent Metropolis kernel $K$. Recall that by construction the stationary distribution of $Q$ is $\pi$, such that
Theorem~\ref{2pertbnd} yields
\[
    W_d(\pi,\pi_K) \leq 
    \frac{2 \sigma^2 e^{\sigma^2} \mu(\widetilde{\Omega})}{\alpha} \Big(1+\frac{1}{\alpha}\Big). 
\]
\end{example}

In summary, the simple scenario above develops that $K$ inherits the TV contraction property \eqref{IneqContTV} from $Q$, with all its benefits, though with a
worse constant. More broadly, we might hope that $K$ inherits any ``good" contractive
property of $Q$ under some small-perturbation condition. 

Less restrictive and more involved estimates with respect to the Wasserstein distance can be found in \cite[Section~2.3]{EberleMaj19}. 
In general the statistics literature has many ``inheritance" properties similar to \eqref{IneqContTV}, see e.g. \cite[Section~3]{medina2016stability} or \cite[Theorem 3.12]{negrea_rosenthal_2021}. One distinction of recent work in computer science is that they focus on inheritance of more abstract quantities under more situation-specific conditions. For example, \cite[Lemma 6]{Chen20} uses a perturbation bound to eventually show that the perturbed kernel satisfies bounds on its conductance profile, given as
\begin{equation*}
    % \sup_{A \in \mathcal{A}} \frac{\int_{x \in A, y \notin A} \mu(dx) Q(x,dy)}{\mu(A)}
    \sup_{A \in \mathcal{A}} \frac{\int_{A}K(x,A^c) \pi_K({\rm d}x)}{\pi_K(A)}
\end{equation*}
for carefully-chosen collections of sets $\mathcal{A}$. Beyond the conductance profile, papers such as \cite{NEURIPS2019_65a99bb7,zhang2023improved} use perturbation results to obtain bounds on log-Sobolev inequalities. A guide to this sort of ``inherited property" analysis is beyond the scope of this chapter; see \cite{zhang2023improved}.

\section{A Survey of Applications} \label{SecExampleAppsGeneric}

We go over the topics from Section \ref{SecTaxonomy} in the same order, with some further details.

\subsection{Subsampling and Approximations for Simple Targets} \label{SecSubsSimpleLong}

One of the most popular applications of perturbation theory is to algorithms that split their datasets - either via subsampling or parallelization. These algorithms implicitly assume that it is possible to approximate the target posterior distribution with a subsample of the full dataset. In the simplest case, such as the product-form distribution in equation \eqref{EqProdPost}, there is a very natural choice - just take a product over a subsample of the datapoints as in equation \eqref{EqSimpleProxy}). Focusing on the simple and overwhelmingly popular case of subsampling methods for product-form targets, a few important lessons have emerged: 

\begin{itemize}
    \item \textbf{Naive Subsampling Suffices for Optimization and Consistent Estimation:} It is not always necessary to run an MCMC chain whose stationary distribution is very close to the nominal target - sometimes a very rough approximation suffices. This situation is very common in the context of using Markov chains for approximate optimization, as in e.g. \cite{SampleOpt19}. In this context, naive subsampling can suffice to get consistent estimates in a weak norm such as Wasserstein \cite{pmlr-v80-desa18a}. The popularity of stochastic gradient Langevin dynamics, which is known to converge to a distribution that is far from the nominal target \cite{Nagapetyan2017}, shows that these ``large-perturbation" chains are very important in practice.
    \item \textbf{Naive Subsampling Does Not Suffice for Accurate Posterior Samples:} If you do want to sample from a stationary measure that is close to the posterior distribution of interest, naive subsampling will typically fail. The details of the failure will depend on the details of the algorithm choice, but typically you will either fail to get a speedup or fail to sample from a distribution that is close to the desired target. See \cite{betancourt15,bardTall17,Nagapetyan2017,NEURIPS2018_335cd1b9,johndrow2020free} for typical examples in very different contexts.
    \item \textbf{Careful Choice of Control Variates Lead to Real Improvement:} Despite the previous point, not all hope is lost. By a combination of careful choice of algorithm and pre-computed control variates, it is possible to obtain both substantial speedups and accurate samples \cite{bardTall17, pass2017, Quiroz19,CVSGMCMC19}.
    \end{itemize}

As a consequence of these three points, most applied research on subsampling MCMC for Bayesian computation focuses on the choice of control variates which we now discuss in some depth and demonstrate using the methodology developed in \cite{Quiroz19}; see \cite{Quiroz18} for a more detailed discussion.

Let $S=\{I_1,\dots,I_k\}$ be a multi set, where $I_j$ are i.i.d. uniform samples from $\{1,\dots,n\}$. Instead of a product-form estimator of \eqref{EqProdPost} as in \eqref{EqNoProxy} or \eqref{EqSimpleProxy}, \cite{Quiroz19} estimates the log-posterior 
\begin{equation} \label{EqSumLogPost}
\log f(\theta | X) = \log f(\theta) +  \sum_{i=1}^{n} \log f(X_{i}|\theta)
\end{equation}
by fixing a collection of control variates $q(X_i|\theta)$ and using a ``sum-form" estimator
\begin{align} \notag
\widehat{\log f(\theta | X)} & = 
\log f(\theta) + q(\theta) + \frac{n}{k}\sum_{j=1}^k (\log f(X_{I_j}|\theta) - q(X_{I_j}|\theta))\\
\label{EqEstSumLogPost}
& =
\log f(\theta) + q(\theta) + \frac{n}{|S|}\sum_{i \in S}(\log f(X_{i}|\theta) - q(X_{i}|\theta)),
\end{align}
where $q(\theta)=\sum_{i=1}^n q(X_i|\theta)$. The third term in \eqref{EqEstSumLogPost} is the simple random sample unbiased estimator of $\log f(\theta|X)-q(\theta)$, and hence \eqref{EqEstSumLogPost} is an unbiased estimator of \eqref{EqSumLogPost}. Note that $q(X_i|\theta)$ should be chosen so that $q(\theta)$ can be evaluated at a cost less than $\mathcal{O}(n)$, otherwise, the estimator in \eqref{EqEstSumLogPost} defeats the purpose of subsampling. \cite{bardTall17} proposes to set $q(X_i|\theta)$ as the second-order Taylor expansion of $\log f(X_i|\theta)$ around the posterior mode and show that $q(\theta)$ can be computed at an $\mathcal{O}(1)$ cost in each MCMC iteration, provided some summary statistics over the full dataset are pre-computed. Thus, the estimator in \eqref{EqEstSumLogPost} has an $\mathcal{O}(|S|)$ cost.

It is well-known that without control variates, i.e., if $q(X_i|\theta)=0$, the simple random sampling approach above is highly inefficient as the $\log f(X_i|\theta)$ vary substantially across the population. The role of $q(X_i|\theta)$ is thus to homogenize the sampling population $d_i(\theta)=\log f(X_{i}|\theta) - q(X_{i}|\theta)$, $i=1,\dots,n$. 
% The variance of
% \eqref{EqEstSumLogPost} is (assuming with replacement sampling)
The variance of $\widehat{\log f(\theta | X)}$ in \eqref{EqEstSumLogPost} is
\begin{equation} \label{EqVarLogLEst}
    \sigma^2_{\widehat{\log f}} = \frac{n^2}{|S|}\sigma_d^2, \,\, \text{where } \sigma_d^2 = \frac{1}{n}\sum_{i=1}^n \left(d_i(\theta) - \overline{d}(\theta)\right)^2 \,\, \text{with } \overline{d}(\theta) = \frac{1}{n}\sum_{i=1}^n d_i(\theta).
\end{equation}
As we will see, keeping this variance sufficiently small is crucial for a successful implementation of subsampling using a pseudomarginal approach. 

\begin{figure}
    \centering    \includegraphics[width=\linewidth]{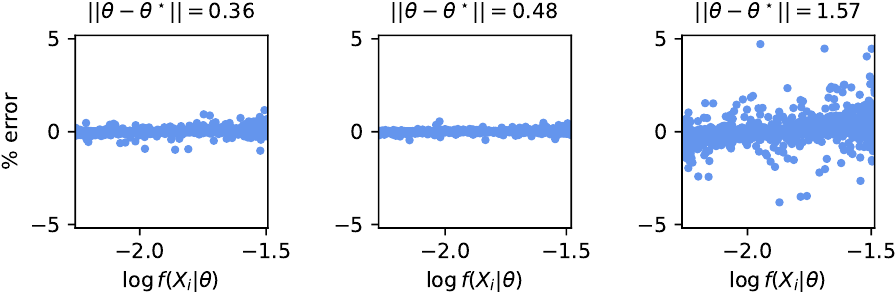}
    \caption{Percentage error for the second-order Taylor control variate approximation $q(X_i|\theta)$ of $\log f(X_i|\theta)$, $i=1,\dots,n$, for three values of $\theta$ in Example \ref{ExSimpleSubsampling}. The $\theta$ values are determined by running a long MCMC chain on the full dataset and using a kernel density estimator to determine the 25\% and 50\% sets of the samples with the highest posterior density. The sample with the largest deviation from the posterior mode $\theta^\star$ within each of the two sets is selected: that of the 25\% and 50\% sets are used in the left and middle panel, respectively. The right panel shows the result using the sample (from the 100\% set) with the largest deviation from $\theta^\star$.} \label{fig:2nd_order_von_Mises_control_variates}
\end{figure}

Assuming that 
$\widehat{\log f(\theta | X)}\sim \mathcal{N}(\log f(\theta|X),\sigma^2_{\widehat{\log f}})$ for \eqref{EqEstSumLogPost},
which can be justified by the central limit theorem, we can use the properties of log-normal variables to show that
\begin{equation} \label{EqUnbiasedEst}
 \exp\left(\widehat{\log f(\theta|X)} - \frac{\sigma^2_{\widehat{\log f}}}{2}  \right)
\end{equation}
is an unbiased estimator of \eqref{EqProdPost}, an idea first explored in \cite{ceperley1999penalty}. There are at least two objections to this analysis. First, 
$\widehat{\log f(\theta | X)}$
is not necessarily normal. Second, and more crucially, computing $\sigma^2_{\widehat{\log f}}$ requires the full dataset (see \eqref{EqVarLogLEst}), hence defeating the purpose of subsampling. \cite{Quiroz19} proposes a pseudomarginal algorithm that uses the likelihood estimator in \eqref{EqUnbiasedEst} with the second-order Taylor control variates expanded around the posterior mode, however, with an estimated variance $\widehat{\sigma}^2_{\widehat{\log f}}$ (computed from the subsample only). Since the resulting estimator in \eqref{EqUnbiasedEst} is not unbiased for \eqref{EqProdPost}, the pseudomarginal does not produce samples from $f(\theta|X)$. \cite{Quiroz19} shows that the algorithm produces samples from a perturbed posterior $\overline{f}(\theta|X)$ which is within $\mathcal{O}(|S|^{-2}n^{-1})$ of $f(\theta|X)$ in total variation norm. For example, in an asymptotic setting where $|S| = \mathcal{O}(\sqrt{n})$, the perturbation error is $\mathcal{O}(n^{-2})$. 

\begin{figure}
    \centering    \includegraphics[width=\linewidth]{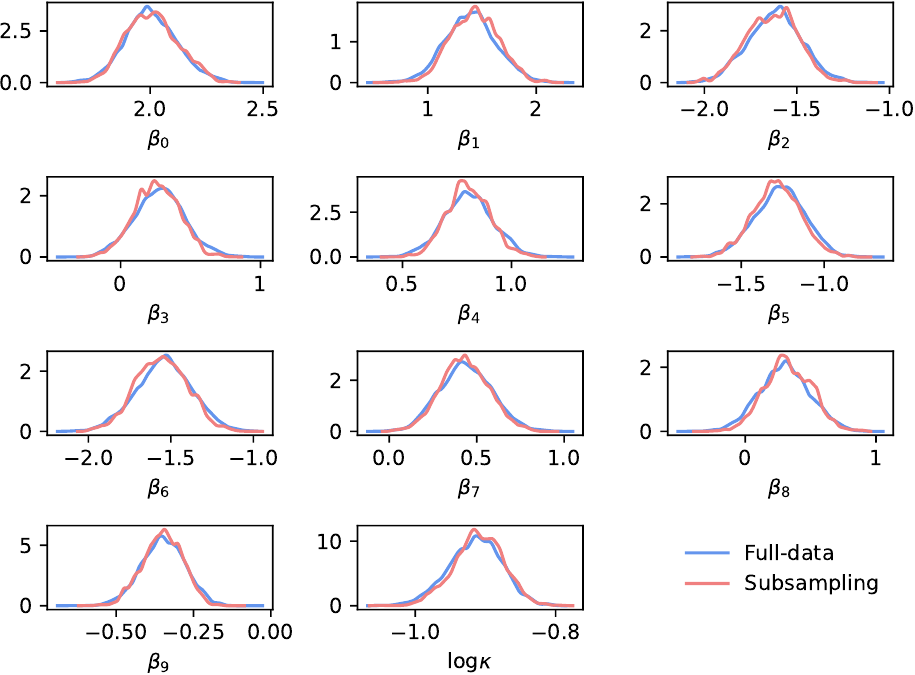}
    \caption{Kernel density estimates of marginal posteriors in Example \ref{ExSimpleSubsampling} with and without subsampling.}
    \label{fig:2nd_order_von_Mises_KDE}
\end{figure}

A key issue when implementing pseudomarginal algorithms is to avoid ``sticky" behaviour of the Markov chain, as this inflates the asymptotic variance. A factor that affects the asymptotic variance is how precise the likelihood estimator is, usually measured via the variance of the log of the likelihood estimator \cite{pitt2012some,doucet2015efficient, sherlock2015efficiency, schmon2021large}. A tuning parameter that controls this variance is the number of ``particles"  used in the estimator (this corresponds to the number of subsamples $k$ used in subsampling). In general, a smaller number of particles, while decreasing the computational cost of each step, increases the variance of the log of the likelihood estimator and may increase the asymptotic variance of the MCMC chain. On the other hand, a larger number of particles, while computationally less efficient, decreases the variance of the log of the likelihood estimator which in turn may decrease the asymptotic variance. A natural question arises: Is it possible to choose the number of particles to achieve optimal performance (in some sense)? Attempts to answer these questions have arrived at different conclusions, depending on the assumptions made. \cite{pitt2012some,doucet2015efficient,sherlock2015efficiency, schmon2021large} suggest that the number of particles should target a variance of the log of the likelihood estimator between 1 and 3.3, resulting in acceptance probabilities in the interval 0.07-0.50. \cite{sherlock2017pseudo} conclude under different assumptions that it is often better to choose a very small number of particles, unless there is a start-up cost. \cite{Quiroz19} follows the approach in \cite{pitt2012some} and suggests the heuristics to set a subsample size that results in $\sigma^2_{\widehat{\log f}} \approx 1$, which in many applications (including Example \ref{ExSimpleSubsampling}) has proven useful.

The role of the control variates is now clear: they determine the factor $\sigma^2_d$ in \eqref{EqVarLogLEst} and thereby $\sigma^2_{\widehat{\log f}}$ --- making it possible to achieve the variance around 1 using a smaller subsample size the more accurate they are. Without any control variates, $\sigma^2_d = \mathcal{O}(1)$ and $\sigma^2_{\widehat{\log f}}=\mathcal{O}(n^2|S|^{-1})$, meaning the subsample size must scale quadratically to keep the variance bounded. \cite{Quiroz19} shows that with the second-order Taylor expanded control variates, $\sigma^2_{\widehat{\log f}}=\mathcal{O}(n^{-1}|S|^{-1})$ and thus scale extremely well for large datasets. 

The next example illustrates the significance of accurate control variates for successful implementation of subsampling. To challenge the methodology, our experiments consider a reasonably small $n$ below (otherwise the variance of the log-likelihood estimator is negligible). The example shows that a control variate based on a first-order Taylor expansion is not accurate enough (given the same subsample size) and results in a sticky pseudomarginal chain.

\begin{figure}
    \centering    \includegraphics[width=\linewidth]{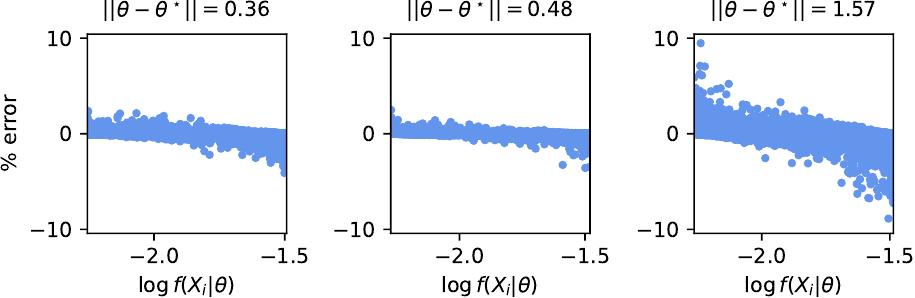}
    \caption{Percentage error for the first-order Taylor control variate approximation $q(X_i|\theta)$ of $\log f(X_i|\theta)$, $i=1,\dots,n$, for three values of $\theta$ in Example \ref{ExSimpleSubsampling}. See the caption of Figure \ref{fig:2nd_order_von_Mises_control_variates} for details.} \label{fig:1st_order_von_Mises_control_variates}
\end{figure}
\newpage

\begin{example}[Subsampling for simple targets] \label{ExSimpleSubsampling}
We consider the following regression model for angular/directional data
\begin{eqnarray*}
 X_{i} | \mu_i, \kappa & \sim & \mathrm{von Mises}\left(\mu_i, \kappa\right), i = 1,\dots, n,\\
 \mu_i & = & 2\arctan\left(\beta^\top z_i\right) \\
 \beta & \sim & \mathcal{N}(0, 10^2 I) \\
 \log \kappa & \sim & \mathcal{N}(0, 10^2),
\end{eqnarray*}
where $I$ denotes the identity matrix and the density of $\mathrm{von Mises}\left(\mu, \kappa\right)$ is given by
$$f(x| \mu, \kappa) = \frac{\exp\left(\kappa \cos(x - \mu)\right)}{2\pi I_0(\kappa)}, \quad -\pi < x < \pi.$$
The parameters in this distribution are the main direction $\mu \in (-\pi, \pi)$ and a measure of concentration $\kappa \in \mathbb{R}$. The normalizing constant $I_0$ is the modified Bessel function of order $0$. With $\theta = (\beta, \log \kappa)^\top$, the log-density is 
\begin{eqnarray*}
 \log f(X_i | \theta) & = & \kappa \cos\left(X_i - 2\arctan\left(\beta^\top z_i\right)\right) - \log(2\pi) - \log I_0(\kappa).
\end{eqnarray*}

We simulate $n=10{,}000$ observations from the model above with $10$ features and construct the control variates $q(X_i|\theta)$ based on a second-order Taylor expansion around the posterior mode $\theta^\star$. The accuracy of the control variates depends on which $\theta$ they are evaluated for. Figure \ref{fig:2nd_order_von_Mises_control_variates} shows the percentage error for three different $\theta$ values, see the caption for details. Figure \ref{fig:2nd_order_von_Mises_KDE} shows the resulting kernel density estimates of the marginal posterior densities using a subsample size of $|S|=100$, which results in estimated values of $\sigma^2_{\widehat{\log f}}$ over the pseudomarginal iterates with mean $\approx 1.5$, median $\approx 0.75$ and max $\approx 21$ (recall that $\approx 1$ is the target). We conclude that subsampling is successful in this application with the given settings and report relative efficiencies between $45{-}68$ (depending on which parameter is estimated) compared to an MCMC using the full-data. The efficiency measure includes the cost of computing the estimator and the asymptotic variance of the MCMC chain; see \cite{Quiroz19} for details.

To illustrate the importance of accurate control variates, let us consider the same settings as above, however, with control variates based on a first-order Taylor expansion. Figure \ref{fig:1st_order_von_Mises_control_variates} shows the percentage error for the same $\theta$ values as in Figure \ref{fig:2nd_order_von_Mises_control_variates}. Figure \ref{fig:1st_order_von_Mises_trace_plots} shows that subsampling is not successful with the less accurate control variates (given the same settings as above). The estimated values of $\sigma^2_{\widehat{\log f}}$ over the pseudomarginal iterates have mean $\approx 11$, median $\approx 7$ and max $\approx 34$, causing the pseudomarginal chain to get stuck. 
\end{example}

\begin{figure}
    \centering    \includegraphics[width=\linewidth]{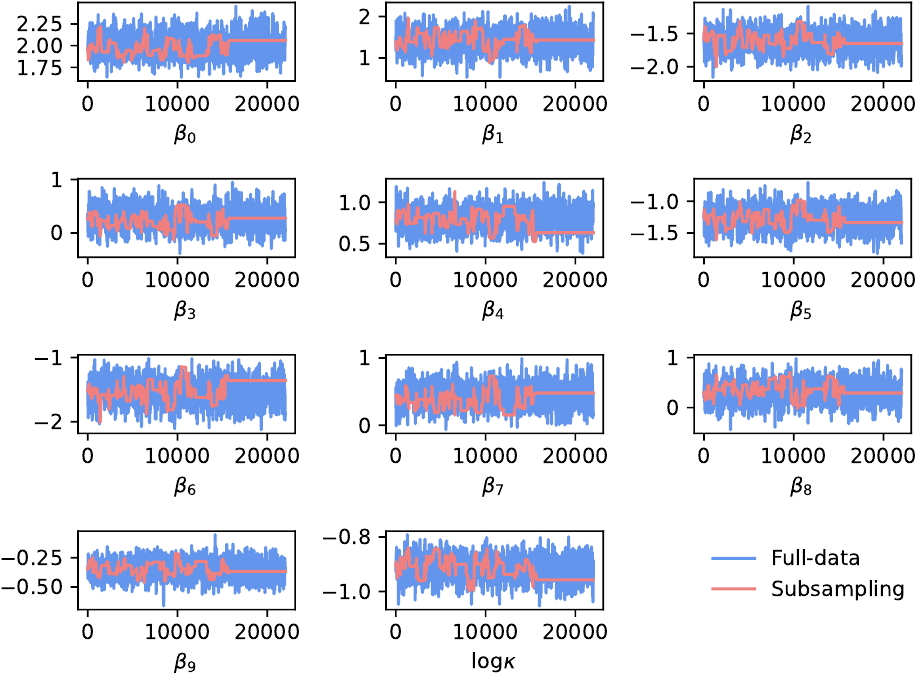}
    \caption{Trace plots in Example \ref{ExSimpleSubsampling} with and without subsampling when using control variates based on a first-order Taylor expansion.}
    \label{fig:1st_order_von_Mises_trace_plots}
\end{figure}

We make some final remarks on the example above. First, the obvious way to prevent the pseudomarginal chain from getting stuck is to increase the subsample size (so that the variance of the estimator is around 1 as discussed above). Second, correlated pseudomarginal methods can be used to allow for a larger variance of the estimator without getting the pseudomarginal chain stuck \cite{tran2016block, deligiannidis2018correlated}. Third, improved grouped control variates have been proposed in \cite{salomone2020spectral}. Fourth, in real large data applications, $n$ is substantially larger, and then the control variates become very accurate. We did not implement any of the above as the purpose of the example is to illustrate the effect of accurate control variates relative to inaccurate control variates.

We call attention to a recently-developed approach to control variates that may be less familiar to a statistical audience than the approaches in \cite{Quiroz19,CVSGMCMC19}. Very informally, the \textit{coreset} approach introduced in \cite{Core16} attempts to approximate  posterior distributions of the form in \eqref{EqProdPost} by finding a small set $S \subseteq \{1,2,\ldots,n\}$ and weights $\{w_{i}\}_{i \in S}$ such that 
\begin{equation} 
f(\theta | X) = f(\theta) \prod_{i=1}^{n} f(X_{i}|\theta) \approx f(\theta) \prod_{i \in S } f(X_{i}|\theta)^{w_{i}}.
\end{equation}
The authors suggest that such an approximation is taking advantage of the fact that large datasets are \textit{redundant.} In a more familiar language, they are using the original single-term likelihoods as $f(X_{i} | \theta)$ as control variates.

The main result of \cite{Core16} stated that one could take a coreset $S$ of size on the order of $\log(n)$ and still obtain good samples. The appeal here is obvious - once you have found $S$ and $\{w_{i}\}_{i \in S}$, running MCMC is much cheaper than even the most aggressive of the random-subsample methods. Of course, these advantages are not free. While the control variates in \cite{Quiroz19,CVSGMCMC19} can be implemented very easily, finding good coresets is very challenging. At the time of writing, no general-purpose methods were known, and efficient computation may require substantial problem-specific adjustments. See Section 3 of the survey \cite{winter2023machine} for descriptions of the state-of-the art for coreset construction methods based on subsampling, sparse regression and variational inference. 

Finally, we note that almost all of the above analysis takes place in the ``standard" regime, where all components of the posterior are somewhat sensitive to all datapoints. In this regime, we've seen that subsampling is quite delicate. However, many important problems are \textit{not} in this regime. For example, in the case of variable selection for high-dimensional Bayesian regression, it is often possible to determine that some coefficients should be exactly 0 with overwhelming probability using only a tiny fraction of the data (see e.g. \cite{screening08} for an overview of this ``screening" problem). In this case, much sparser subsamples can be sufficient to characterize some features of the posterior, and consequently much greater speedups are possible. See \cite{SSLASSO,JMLR:v21:19-536,sparseBayes21} for details on how these observations can be turned into concrete MCMC algorithms.

\subsection{Parallelization and Divide-and-Conquer}

The theory for parallel algorithms is related to the theory for subsampling described in Section \ref{SecSubsSimpleLong}, but more difficult. The most important difference is that parallel algorithms do \textit{not} require subsample-based posteriors (such as equation \eqref{EqSimpleProxy}) to be very close to the full-data posterior (such as \eqref{EqProdPost}). Indeed, in typical applications the subsample-based posteriors associated with different subsamples will have nearly-disjoint support  (see e.g. \cite{Wang2013ParallelizingMV}). 

This results in two fundamentally different ways to parallelize MCMC algorithms based on subsampling. The first approach, called split-and-merge algorithms is conceptually simplest: run MCMC algorithms independently for each subsample, then merge the results. The main difficulty here comes from the fact that the subsamples associated with different datasets will often disagree quite sharply, leading to poor performance. Clever merging strategies can mitigate this difficulty in specific situations \cite{Wang2013ParallelizingMV,neiswanger2013asymptotically,SriDivide18,scott2022bayes, nemeth2018merging}, but few guarantees are available for practical algorithms.

The second choice is difficult to implement but results in controlled error: run one large MCMC algorithm across all cores, designed in such a way that only a small amount of communication is required. Some tricks allow for parallel algorithms that exactly preserve the stationary distribution, but these are often difficult-to-implement or specialized \cite{GPUAccPar19} or require substantial amounts of communication between cores \cite{CalderPar14}. Another option is to give a more generic algorithm that introduces a small but manageable amount of error. This typically involves introducing auxiliary variables that relate the different subsamples, either via an explicit model \cite{RendelGCMC21,pmlr-v151-de-souza22a} or implicitly via an optimization procedure \cite{VCM15,EPWay2020}. As with the other parallel algorithms, few careful analyses are available - though see \cite{Splitting22,pmlr-v139-plassier21a}. 

Although there is a large literature on parallel MCMC, at the time of writing its theoretical properties remain poorly understood. Perhaps the most important lesson for theorists is to recognize that much of the guidance in Section \ref{SecSubsSimpleLong} does \textit{not} seem to apply to parallel MCMC algorithms, even though the situations appear to be similar. In the simple case of split-and-merge algorithms, the individual subposteriors do \textit{not} need to be close to each other (or the full posterior), and so it is \textit{not} necessary to have Markov chains that are small perturbations of the large chain (see e.g. the numerical studies in \cite{Mishchenko2022AsynchronousSB}). In the same vein, there does not seem to be any reason to believe that the component Markov chains of more sophisticated parallel MCMC algorithms should be small perturbations of any full-data MCMC algorithm.

\subsection{Subsampling and Approximation for Complex Targets}

Although there are several examples \cite{rastelli2023computationally, bradley2021approach, saha2023incorporating}, we are not aware of any general techniques for defining or analyzing perturbed Markov chains in this regime. Finding such bounds has been mentioned as a central open problem in \cite{winter2023machine,johndrow2020free}. 

Approximations of complex targets may sometimes result in a posterior for which subsampling is easier, such as a ``product-form" posterior. Examples using the Whittle log-likelihood approximation \cite{whittle1953estimation} of the multivariate Gaussian log-likelihood to subsample stationary time series models are found in \cite{salomone2020spectral, VILLANI2022}. The multivariate Gaussian log-likelihood is expensive to compute: despite the Toeplitz structure of the covariance matrix, the cost scales as the square of the length of the time series. The Whittle log-likelihood approximation utilises asymptotic properties of the discrete Fourier transform (DFT) of the data to form a likelihood in the frequency domain. The DFT at the different frequencies are asymptotically independent, yielding a ``product-form" structure for the resulting likelihood.

\subsection{Inverse Problems from PDEs and Other Expensive Models}\label{SecInvDetail}

The next important class of problems are inverse problems from differential equations. Abstractly, these are problems in which you can observe some features of the solution of a PDE and wish to infer parameters of the PDE model. We give a very short treatment of the subject and focus on perturbations; see Chapter 16 of \cite{brooks2011handbook}, Chapter 11 of this book, and \cite{Dashti2017} for important techniques that are not related to perturbations.

To give a concrete example, we consider Example 1.3 of \cite{Dashti2017}. Here we fix a bounded open set $D \subset \mathbb{R}^{d}$ with boundary $\partial D$ and consider the PDE:
\begin{eqnarray}\label{EqPDEModel}
- \nabla(\kappa \nabla g) = f, \quad x \in D \\
g = 0, \qquad x \in \partial D,
\end{eqnarray}
where $f, \kappa \, : \, \mathbb{R}^{d} \to \mathbb{R}$ are functions. Denote by $g_{\kappa}$ the solution associated with a particular choice of $\kappa$. 

To turn this from a model into a statistical problem, we might imagine a situation in which $f$ is known, we can make noisy observations $y_{j} \sim \mathcal{N}(g(x_{j}),1)$ of the solution $g$ at a finite collection of points $x_{1},\ldots,x_{n}$, and we wish to infer the entire parameter field $\kappa$. In typical mathematical notation, we wish to sample from the posterior $p(\kappa | y)$. More complex models with a similar flavour are used in e.g. inferring information about subsurface hydrology from surface measurements \cite{Subwater10}.

Equation \eqref{EqPDEModel} presents some obvious problems. The most immediate is beyond the scope of this chapter: it is enormously ill-posed, as a finite number of observations can't ``pin down" all of $\kappa$. See the advice in \cite{Dashti2017} for how to deal with ill-posedness. More relevantly for this chapter, doing MCMC requires solving the PDE \eqref{EqPDEModel} for each candidate field $\kappa$. Solving PDEs can be computationally expensive, and there is no obvious way to make this cheaper by e.g. subsampling the observations.

There are two simple strategies for dealing with this computational issue: making the PDE easier to solve and solving it less frequently. The first of these strategies is, fortunately, the subject of an enormous literature. The simplest approach is to discretize the domain $D$ and solve the associated finite-difference equation (see e.g. Chapter 16 of \cite{brooks2011handbook}). It is also possible to solve a regularized version of the full infinite-dimensional problem; see e.g. \cite{Dashti2017}. Finally, multilevel Monte Carlo (MLMC) techniques allow one to use inexpensive low-resolution approximations to $p$ at \textit{most} stages while adaptively using a small number of expensive high-resolution approximations to retain a low overall bias \cite{Giles2013MultilevelMC}. 

With the exception of MLMC,\footnote{or sophisticated adaptive discretization schemes, such as \cite{Bigoni_2020}.} all of these approaches effectively involve replacing one model with a different model that is kept throughout an MCMC run. This means that you can do a perturbation-based analysis directly on the model in equation \eqref{EqPDEModel} \textit{independently} of the choice of MCMC algorithm. Therefore, the resulting theory does not use results such as Theorem \ref{1th:Z_m} or more sophisticated versions as in \cite{johndrow2017error,SchweizerPerturb18,negrea_rosenthal_2021}.

The basic idea behind the second strategy is to solve equation \eqref{EqPDEModel} for some values of $\kappa_{1},\ldots,\kappa_{N}$, then use $\{p(\kappa_{j} | y)\}_{j=1}^{N}$ to fit a model $\hat{p}(\kappa)$ that estimates $p(\kappa | y)$. In principle one could choose $\kappa_{1},\ldots,\kappa_{N}$ \textit{before} running the Markov chain, in which case one could again compute bounds on the error $\hat{p}(\kappa) - p(\kappa|y)$ that applied to all MCMC algorithms. In practice this is unlikely to work well in even moderate-dimensional problems, because it is difficult to choose a set of reference points $\{\kappa_{j}\}_{j=1}^{N}$ that cover the posterior distribution well. As a consequence, it is typically better to choose points \textit{during} the run of an MCMC algorithm (see \cite{Conrad16}). This means that any analysis must combine perturbation bounds such as Theorem \ref{1th:Z_m} with so-called \textit{adaptive MCMC} bounds (see Chapter 4 of \cite{brooks2011handbook} and Chapter 2 of this book).

Fortunately, these approaches combine well, and near-optimal convergence rates are known \cite{davis2022rate}. This theoretical work gives two practical heuristics:
\begin{enumerate}
    \item \textbf{Number of Evaluations:} Choose the number of model evaluations $N=N(T)$ by time step $T$ of the Markov chain so that the (i) model error $\hat{p}(\kappa) - p(\kappa|y)$ and (ii) Monte Carlo error are roughly the same size.
    \item \textbf{Bias and Ergodicity:} As noted in \cite{Conrad16}, natural estimators $\hat{p}$ may \textit{not} result in ergodic Markov chains. To fix this problem, one can introduce some small global bias towards a reference point $\kappa_{0}$, replacing e.g. $\hat{p}(\kappa)$ by $\hat{p}(\kappa)e^{-C_{T} \|\kappa-\kappa_{0}\|^{2}}$ for some sequence $C_{T} \rightarrow 0$. As long as $C_{T}$ is chosen consistency with the heuristic in (1), this introduces negligible Monte Carlo error and stabilizes the resulting algorithm.
\end{enumerate}

Note that these two approximation strategies work on different parts of the problem and can be combined. The first strategy essentially assumes that the solution $p$ of the PDE is insensitive to small perturbations, while the second assumes that the likelihood evaluated at a parameter $\kappa$ is insensitive to small perturbations.

\subsection{Simulation-Based Methods}

In the context of MCMC using approximate Bayesian computation, one has access to an explicit formula for the stationary distribution as a perturbation of the nominal target. In most situations that we are aware of, it is easier to analyze this perturbation formula directly than to propagate error bounds for the perturbed transition kernel. As such, we suggest interested readers look at specialized literature, such as \cite{Sisson2018handbook}, for practical advice on approximate Bayesian computation and related simulation-based methods.

\subsection{Doubly-Intractable Distributions}

See Chapter 16 of this book for practical advice on algorithms for sampling from doubly-intractable distributions.

Although several algorithms for sampling from doubly-intractable distributions can be viewed as perturbations, we are not aware of a substantial perturbation-theoretic literature on the subject.
Let us mention that a related stability analysis of doubly-intractable distributions appears in \cite{habeck2020stability}. Furthermore, \cite{MEDINAAGUAYO20202200} shows that a perturbation theoretic MCMC analysis is in principle possible for plain iid Monte Carlo within Metropolis, though some care is required. However, for more commonly used state-of-the-art algorithms it is not at all obvious how to view them as ``small" perturbations. 

As pointed out in Section \ref{SecTaxonomy}, a non-linear transformation of an unbiased estimator of the normalization constant cannot provide a nonnegative unbiased estimator of the posterior distribution under feasible assumptions. However, it is possible to compute an unbiased estimator of the posterior distribution that is occasionally negative. \cite{Roulette15} suggests to carry out a pseudomarginal algorithm targeting the absolute value of estimator, thus resulting in a perturbed posterior. \cite{quiroz2021block} terms this the \textit{signed} pseudomarginal approach. This perturbed posterior is hard to analyze due to the Russian roulette type of estimators based on random truncation. However, the samples may be re-weighted to obtain consistent estimates of expectations with respect to the doubly-intractable unperturbed posterior, avoiding the need of a perturbation analysis. The asymptotic variance of the signed pseudomarginal has two components: the variance of the log of the likelihood estimator and the probability of a nonnegative estimator. The expressions needed for the optimal tuning proposed in \cite{quiroz2021block} are intractable for the Russian roulette estimators used in \cite{Roulette15}. \cite{yang2022} proposes a more tractable estimator for doubly-intractable problems, and derive heuristic guidelines for optimally tuning the estimator. Moreover, \cite{yang2022} proposes a correlated pseudomarginal version of their algorithm.

\section{Open Questions}
We give some personal favourite open questions related to perturbations of Markov chains and MCMC analysis.

\subsection{Comparison of Algorithms}

 \textbf{Peskun-like or Computer-Aided Orders:} One of the main purposes of MCMC theory is to provide users with some guidance on which algorithm to use in a given situation. Unfortunately, pen-and-paper error estimates for individual MCMC algorithms are almost always so rough as to be useless for this purpose \cite{jones2001honest} (though see e.g. \cite{jacob2020unbiased,CFTP96} for computer-aided estimates that can be substantially sharper). The main theoretical results that give such guidance come from directly comparing two chains, as in Peskun ordering and various followups (see the survey \cite{MiraOrderingSurvey01} as well as recent extensions \cite{andrieu2022comparison,PeskNonrev21,CompExactApprox16} using other techniques). At the moment, there are very few results allowing direct comparison of approximate MCMC algorithms, and these results only compare chains with the same bias \cite{Bornn2014TheUO,sherlock2017pseudo}. Are other comparisons possible? 

\textbf{Nearly-Sharp Analyses:} While gradient and stochastic-gradient methods are very popular, almost no theoretical results are precise enough to allow precise comparisons of algorithms or hyperparameter tuning. This suggests an enormous number of open problems: choose almost any pair of similar methods, and see if you can show that one is better than another under reasonable circumstances. See e.g.  Section 2.4 of \cite{SHMCMCSurvey} for well-supported conjectures. See \cite{chen2023does} for the one recent example of a successful comparison of two popular algorithms.

 \subsection{Theory of Large Perturbations }
In some cases, perturbed Markov chains give good answers even when the size of the perturbation is much larger than bounds like inequality \eqref{eq: 1st_pert_stat} would allow. At the moment, this phenomenon is only understood for a small number of very simple cases \cite{lin2023perturbation}. The following are two settings that are of particular practical interest:
 
\textbf{Doubly-Intractable Distributions:} Published analyses of MCMC algorithms for doubly-intractable algorithms do not generally make use of perturbation theory. The recent empirical results in \cite{kang2023measuring} provide both a possible explanation and a possible route forward. The paper shows that a certain  approximation to the exchange algorithm has a stationary distribution that is very close to the desired target \textit{even in a regime} where the internal approximation is extremely far from convergence. Because the perturbations are large, we don't expect classical perturbation results to apply directly in this context. This leaves the open question: what phenomenon allows the stationary distribution to be close to the target in this situation? 

\textbf{Parallel Algorithms:} Divide-and-conquer methods are a natural alternative to subsampling and other ``noisy" MCMC methods. While some divide-and-conquer methods focus on combining MCMC runs after they have completed \cite{scott2022bayes,Wang2013ParallelizingMV,VynerSwiss23}, many others allow some communication between computers during MCMC runs  \cite{Conrad2016ParallelLA,EPWay2020,RendelGCMC21}. There is widespread perception that a small amount of communication will improve convergence. However, the very feature that makes these algorithms better also makes them harder to analyze with perturbation-based methods: the basic algorithms run on each computer are generally \textit{not} small perturbations of any ideal global algorithm. How can perturbation bounds be combined with other methods, such as those in the variational inference literature, to analyze such limited-communication algorithms? See \cite{Splitting22,VynerSwiss23} for some recent related work.

\subsection{Efficiency of The Pseudomarginal Trick} 

Many perturbed MCMC algorithms arise from replacing an expensive-to-compute quantity by a computationally-cheaper stochastic approximation. In these cases, one can use a plug-in estimator directly and keep the original state space, or use the pseudomarginal trick to obtain a closely-related chain on a larger state space (see e.g. \cite{Quiroz19}). The pseudomarginal algorithm is typically easier to analyze, as one has access to the stationary distribution. However, the following question is still largely open: do these pseudomarginal versions of perturbed algorithms actually have smaller error, or is it merely the case that existing theoretical results give better error bounds? 

\subsection{Optimal Complexity of Subsampling Algorithms} 
At the time of writing, there are almost no sharp analyses on the amount of computational effort required to obtain good samples. We give two informal problems in this area. For a fixed MCMC algorithm, denote by $\tau(n,\epsilon)$ the computational effort needed to get a 
Monte Carlo sample with total root-mean-square error of less than $\epsilon$ for a dataset with $n$ datapoints. The idea behind all approximate MCMC algorithms is to reduce $\tau(n,\epsilon)$, but the following questions about how to do this remain open:

\textbf{Easy Targets:} Consider the logistic regression model, and more generally sufficiently-nice models of the form \eqref{EqProdPost}. It is straightforward to check that $\tau(n,\epsilon) \lesssim n \epsilon^{-2}$ for typical Metropolis-Hastings algorithm, and \cite{Nagapetyan2017,johndrow2020free} show the matching bound $\tau(n,\epsilon) \gtrsim n \epsilon^{-2}$ for stochastic gradient Langevin dynamics and a broad range of reversible chains. Is it possible to break at least the lower bound's dependence on $\epsilon$ by choosing a more clever nonreversible chain? Work on subsampling for continuous-time nonreversible algorithms such as \cite{pmlr-v119-bierkens20a} suggest that the answer might be yes, but to date precise comparisons remain difficult. 

\textbf{Complex Targets:} Are there any natural families of models where $\tau(n,\epsilon)$ grows super-linearly in $n$ for all algorithms? How would one check such a thing? See related open questions in Section 3 of \cite{winter2023machine}.

\subsection{Slice Sampling Under Perturbation}
The slice sampling paradigm offers an efficient way for the construction of a Markov chain that avoids local random walk behavior, requires (almost) no tuning and usually offers a richer choice of possible updates, for details check \cite{neal2003slice} and \cite{murray2010elliptical}.
However, the stability regarding evaluations of the target density is so far not explored, even though the methodology is heavily used 
as black box sampling scheme.

\subsection{Perturbed Quasi-Stationary Markov Chains}
   Recently, killed Markov chains with quasi-stationary distribution have been used for approximate sampling of pre-specified target distributions to perform Bayesian inference for large data sets \cite{pollock2020quasi}. Truncating the killing rate lowers the computational cost and yields a perturbed Markov chain. In \cite[Proposition~1]{rudolf2021perturbation} a preliminary discrete time perturbation result is provided in the approximate killing rate setting. However, this turns out to be of rather theoretical value. In particular an extension to the Wasserstein distance and a ``more general" quantity to measure the difference of the perturbed and unperturbed transition mechanism would be of interest.

\subsection{Improved Control Variates} 
As demonstrated in Section \ref{SecSubsSimpleLong}, control variates are crucial for a successful implementation of subsampling. The appropriateness of a second-order Taylor expanded control variate depends on to which extent the log-density is locally quadratic. \cite{salomone2020spectral} proposes to divide observations into groups and approximate the sum of the log-densities within a group using a second-order Taylor expansion. The Bernstein von Mises theorem (asymptotic normality of the posterior) suggest a quadratic form of the sum of log-densities (assuming regularity conditions), given enough observations in the group. However, for models with highly complex log-densities, the approximation may be poor even after grouping. Exploring control variates that are flexible enough, while still having an $\mathcal{O}(|S|)$ cost, remains an open problem.

\bibliographystyle{plain} 
\let\clearpage\relax % Prevent LaTeX from inserting a page break
\bibliography{biblio}

\end{document}